\documentclass[envcountsame,envcountsect]{llncs}

\usepackage{graphicx}
\usepackage{amssymb}
\usepackage{epstopdf}
\usepackage{amsmath}
\usepackage{qip}
\usepackage{framed}
\usepackage{epsfig}
\usepackage{subfigure}
\usepackage{mathrsfs}

\usepackage{wrapfig}

\usepackage{xspace}  

\newcommand{\fullv}[2]{#1}  
\newcommand{\delete}[1]{} 

\newcommand{\cnot}{{\sf CNOT}}

\newcommand{\hgate}{{\sf H}}
\newcommand{\pgate}{{\sf P}}
\newcommand{\rgate}{{\sf R}}

\newcommand{\ad}[1]{\ensuremath{\tilde{\mathscr{#1}}}}
\newcommand{\ada}{\ensuremath{\ad{A}}}

\newcommand{\adb}{\ensuremath{\ad{B}}}

\newcommand{\adbt}[1]{\ensuremath{\ad{B}_t}}

\newcommand{\xprot}[1]{\ensuremath{{\sf #1}}}

\newcommand{\aprot}{\ensuremath{\mathscr{A}}}
\newcommand{\aprots}{\ensuremath{\aprot^*}}

\newcommand{\bprot}{\ensuremath{\mathscr{B}}}
\newcommand{\bprots}{\ensuremath{\bprot^*}}

\newcommand{\phiout}{\ensuremath{\ket{\phi_{{\tt out}}}}}

\newcommand{\phiin}{\ensuremath{\ket{\phi_{{\tt in}}}}}

\newcommand{\compose}{\circledast}

\newcommand{\pu}[1]{\ensuremath{\Pi_{#1}}}
\newcommand{\puq}[1]{\ensuremath{\pu{#1}^{\qop{O}}}}

\newcommand{\view}[3]{\ensuremath{{\nu}_{#3}(#1,#2)}}

\newcommand{\simul}[1]{\ensuremath{\mathscr{S}_{#1}}}
\newcommand{\simulg}[1]{\ensuremath{\mathscr{S}^*_{#1}}}
\newcommand{\simulator}[1]{\ensuremath{\mathscr{S}}(#1)}

\input{Qcircuit}
\newcommand{\twowire}[1]{\ar@{}[d]_-{\txt{#1}\Bigg\{}}
\newcommand{\twowirer}[1]{\ar@{}[d]^-{\Bigg\}\txt{#1}}}

\newcommand{\wire}[1]{{\ensuremath{\tt #1}}}
\newcommand{\reg}[1]{\ensuremath{\tt #1}}
\newcommand{\hreg}[1]{\ensuremath{\mathcal{#1}}}
\newcommand{\hdim}[1]{\ensuremath{{\mathcal{H}}_{#1}}}

\newcommand{\ao}[1]{\ensuremath{\hreg{A}^{\qop{O}}_{#1}}}
\newcommand{\bo}[1]{\ensuremath{\hreg{B}^{\qop{O}}_{#1}}}

\newcommand{\nl}{{\sc and-box}}
\newcommand{\pup}[1]{\ensuremath{P_{U}}}
\newcommand{\ug}{\ensuremath{{\cal UG}}}

\newcommand{\pos}[1]{\ensuremath{\mathrm{Pos}(#1)}}
\newcommand{\dens}[1]{\ensuremath{\mathrm{D}(#1)}}
\newcommand{\unit}[1]{\ensuremath{\mathrm{U}(#1)}}
\newcommand{\linop}[1]{\ensuremath{\mathrm{L}(#1)}}
\newcommand{\lin}[2]{\ensuremath{\linop{#1,#2}}}
\newcommand{\alice}[1]{\ensuremath{\mathscr{A}_{#1}}}

\newcommand{\bob}[1]{\ensuremath{\mathscr{B}_{#1}}}

\newcommand{\qop}[1]{\ensuremath{\mathscr{#1}}}

\newcommand{\zo}{\{0,1\}}
\newcommand{\ran}{\stackrel{\$}{\in}}

\newcommand{\hstatei}[2]{\ensuremath{\rho_{#1}(#2)}}

\newcommand{\astatei}[3]{\ensuremath{\tilde{\rho}_{#1}(#2,#3)}}

\newcommand{\trans}[1]{\ensuremath{\mathscr{T}_{#1}}}

\newcommand{\swap}{{\ensuremath{\sf SWAP}}}

\newcommand{\keya}[2]{\ensuremath{K_{\aprots}^{#2}(\wire{#1})}}
\newcommand{\keyb}[2]{\ensuremath{K_{\bprots}^{#2}(\wire{#1})}}
\newcommand{\kxa}[2]{\ensuremath{X_{\aprots}^{#2}(\wire{#1})}}
\newcommand{\kza}[2]{\ensuremath{Z_{\aprots}^{#2}(\wire{#1})}}
\newcommand{\kxb}[2]{\ensuremath{X_{\bprots}^{#2}(\wire{#1})}}
\newcommand{\kzb}[2]{\ensuremath{Z_{\bprots}^{#2}(\wire{#1})}}

\DeclareMathOperator{\Span}{span}

\allowdisplaybreaks

\title{Secure Two-Party Quantum Evaluation of Unitaries Against
Specious Adversaries}
\author{ Fr\'ed\'eric Dupuis\inst{1}\thanks{Supported by Canada's NSERC Postdoctoral Fellowship Program.} \and Jesper Buus Nielsen\inst{2} 
\and Louis Salvail\inst{3}\thanks{Supported by Canada's NSERC discovery grant, 
MITACS, and the  QuantumWorks networks(NSERC).}
}
\institute{ 
 Institute for Theoretical Physics, ETH Zurich, Switzerland\\
  \email{dupuis@phys.ethz.ch} \and
 DAIMI, Aarhus University, Denmark\\
  \email{jbn@cs.au.dk} \and
  Universit\'e de Montr\'eal (DIRO), QC, Canada\\
  \email{salvail@iro.umontreal.ca}}

\begin{document}


\maketitle

\begin{abstract}
  \fullv{We show that any two-party quantum computation, specified by a
  unitary which simultaneously acts on the registers of both parties,
  can be securely implemented against a quantum version of classical
  semi-honest adversaries that we call specious.  

  We first show that no statistically private protocol exists for
  swapping qubits against specious adversaries.  The swap
  functionality is modeled by a unitary transform that is not
  sufficient for universal quantum computation. It means that
  universality is not required in order to obtain impossibility proofs
  in our model. However, the swap transform can easily be implemented
  privately provided a classical bit commitment scheme.
   
  We provide a simple protocol for the evaluation of any unitary
  transform represented by a circuit made out of gates in some
  standard universal set of quantum gates.  All gates except one can
  be implemented securely provided one call to swap made available as
  an ideal functionality.  For each appearance of the remaining gate
  in the circuit, one call to a classical AND-box is required for
  privacy. The AND-box can easily be constructed from oblivious
  transfer. It follows that oblivious transfer is universal for
  private evaluations of unitaries as well as for classical circuits.

  Unlike the ideal swap, AND-boxes are classical primitives and cannot
  be represented by unitary transforms.  It follows that, to some
  extent, this remaining gate is the hard one, like the
  \textsc{and} gate for classical two-party computation.}{
 We describe how any two-party quantum computation, specified by a
  unitary which simultaneously acts on the registers of both parties,
  can be privately implemented against a quantum version of classical
  semi-honest adversaries that we call specious.  Our construction requires
  two ideal functionalities to garantee privacy: a private SWAP between registers held by the two parties
  and a classical private AND-box equivalent to oblivious transfer.  If the unitary to be evaluated
  is in the Clifford group then only one call to SWAP is required for privacy. On the other hand,
  any unitary not in the Clifford requires one call to an AND-box per \rgate{}-gate in the circuit.  
  Since SWAP is itself in the Clifford group, this functionality is universal for the private evaluation
  of any unitary in that group. 
  SWAP can be built from a classical bit commitment scheme or an AND-box but 
  an AND-box cannot be constructed from SWAP.
  It follows that unitaries in the Clifford group are to some extent the easy ones.
  We also show that SWAP cannot be implemented privately in the bare
  model.  
  }
 \end{abstract}



\section{Introduction}\label{intro}

In this paper, we address the problem of privately evaluating some
unitary transform $U$ upon a joint quantum input state held by two
parties.  Since unitaries model what quantum algorithms are
implementing, we can see this problem as a natural extension of secure
two-party evaluation of functions to the quantum realm.  Suppose that
a state $\phiin\in \hreg{A}\otimes\hreg{B}$ is the initial shared
state where Alice holds register \hreg{A} and Bob holds register
\hreg{B}.  Let $U\in \unit{\hreg{A}\otimes \hreg{B}}$ be some unitary
transform acting upon \hreg{A} and \hreg{B}. What cryptographic
assumptions are needed for a private evaluation of $\phiout = U\phiin$
where {\em private} means that each player learns no more than in the
ideal situation depicted in Fig.~\ref{ideal}? Of course, answers to
this question  depend upon the adversary we are
willing to tolerate.
 
 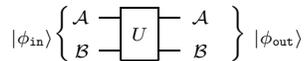
\begin{wrapfigure}{r}{0.333\textwidth}
  \vspace{-20pt}
 \begin{center}
  \mbox{\scriptsize \Qcircuit @C=1em @R=.7em 
  {\twowire{{\phiin}} &  
    \lstick{\hreg{A}} & \multigate{1}{U} & \qw & \hreg{A} &
  \twowirer{  {\ket{\phi_{\tt out}}}}\\
    &  \lstick{\hreg{B}}  & \ghost{U}        & \qw &\hreg{B} &
    }}
\end{center}
\small{\caption{Ideal Functionality for unitary $U$.}
\label{ideal}}
  \vspace{-10pt}
\end{wrapfigure}

 In \cite{SSS09}, it was shown that unitaries cannot be used to
 implement classical cryptographic primitives. Any non-trivial
 primitive implemented by unitaries will necessarily leak information
 toward one party. Moreover, this leakage is available to a weak class
 of adversaries that can be interpreted as the quantum version of
 classical semi-honest adversaries. It follows that quantum two-party
 computation of unitaries cannot be used to implement classical
 cryptographic primitives. This opens the possibility that the
 cryptographic assumptions needed for private evaluations of unitaries
 are weaker than for their classical counterpart. So, what classical
 cryptographic assumptions, if any, are required to achieve privacy in
 our setting? Are there unitaries more difficult to evaluate
 privately than others?
 
 In this work, we answer these questions against a class of weak
 quantum adversaries, called specious, related to classical
 semi-honest adversaries.  We say that a quantum adversary is specious
 if at any step during the execution of a protocol, it can provide a
 judge with some state that, when joined with the state held by the
 honest player, will be indistinguishable from a honest
 interaction. In other words, an adversary is specious if it can pass
 an audit with success at any step.  Most known impossibility proofs
 in quantum cryptography apply when the adversary is restricted to be
 specious. Definitions similar to ours have been proposed for the
 quantum setting and usually named semi-honest. However, translating
 our definition to the classical setting produces a strictly stronger
 class of adversaries than semi-honest\footnote{As an example, assume
 there exist public key cryptosystems where you can sample a public
 key without learning the secret key. Then this is a semi-honest
 oblivious transform: The receiver, with choice bit $c$, samples
 $pk_c$ in the normal way and learns its corresponding secret key and
 samples $pk_{1-c}$ without learning its secret key. He sends
 $(pk_0,pk_1)$.  Then the sender sends $(E_{pk_0}(m_0),E_{pk_1}(m_1))$
 and the receiver decrypts $E_{pk_c}(m_c)$.  This is not secure
 against a specious adversary who can sample $pk_{1-c}$ along with its
 secret key $sk_{1-c}$ and then delete $sk_{1-c}$ before the
 audit.}\fullv{, as demonstrated in Appendix~\ref{app:spec}}{} which
 justifies not adopting the term \emph{semi-honest}. We propose the
 name \emph{specious} as the core of the definition is that the
 adversary must appear to act honestly.

\paragraph{Contributions.}
First, we define two-party protocols for the evaluation of unitaries
having access to oracle calls. This allows us to consider protocols
with security relying on some ideal functionalities in order to be
private. We then say that a protocol is in the {\em bare model} if it
does not involve any call to an ideal functionality.  We then formally
define what we mean by specious adversaries.  Privacy is then defined
via simulation. We say that a protocol for the two-party evaluation of
unitary $U$ is private against specious adversaries if, for any joint
input state and at any step of the protocol, there exists a simulator
that can reproduce the adversary's view having only access to its own
part of the joint input state. Quantum simulation must rely on a
family of simulators for the view of the adversary rather than one
because quantum information does not accumulate but can vanish as the
protocol evolves.  For instance, consider the trivial protocol that
let Alice send her input register to Bob so that he can apply locally
$\phiout=U\phiin$ before returning her register. The final state of
such a protocol is certainly private, as Bob cannot clone Alice's
input and keep a copy, yet at some point Bob had access to Alice's
input thus violating privacy. No simulator can possibly reproduce
Bob's state after he received Alice's register without having access
to her input state.

Second, we show that no  protocol 
can be shown
statistically private against
specious adversaries  
in the bare model for a very simple unitary: the swap gate. As
the name suggests, the swap gate simply permutes Alice's and Bob's
input states.  Intuitively, the reason why this gate is impossible is
that at some point during the execution of such protocol, one party
that still has almost all its own input state receives a non-negligible
amount of information (in the quantum  sense) about the other party's input state.
At this point,
no simulator can possibly re-produce the complete
state held by the receiving party since a call to the ideal functionality 
only provides access to
the other party's state while no call to the ideal functionality only
provides information about that party's own input. 
Therefore, any simulator cannot re-produce a state that contains
information about the input states of both parties. 
It follows that
cryptographic assumptions are needed for the private evaluation of
unitaries against specious adversaries.  On the other hand, a
classical bit commitment is sufficient to implement the swap privately
in our model.

Finally, we give a very simple protocol for the private evaluation of
any unitary based on ideas introduced by \cite{GC99a,GC99b} in the
context of fault tolerant quantum computation. Our construction is
similar to Yao's original construction in the classical
world\cite{Yao86,kilian88}. We represent any unitary $U$ by a quantum
circuit made out of gates taken from the universal set $\ug=\{X,Y,Z,
\cnot, \hgate, \pgate, \rgate\}$ \cite{NC00}.  The protocol evaluates
each gate of the circuit upon shared encrypted input where the
encryption uses the Pauli operators $\{X,Y,Z\}$ together with the
identity.  In addition to the Pauli gates $X,Y$, and $Z$, gates \cnot,
\hgate, and \pgate\ can easily be performed over encrypted states
without losing the ability to decrypt.  Gates of that kind belong to
what is called the {\em Clifford group}.  The \cnot\ gate is the only
gate in \ug\ acting upon more than one qubit while the \rgate-gate is
the only one that does not belong to the Clifford group. In order to
evaluate it over an encrypted state while preserving the ability to
decrypt, we need to rely upon a classical ideal functionality
computing securely an additive sharing for the AND of Alice's and
Bob's input bits.  We call this ideal functionality an AND-box.  Upon
input $x\in\{0,1\}$ for Alice and $y\in\{0,1\}$ for Bob, it produces
$a\in_R \{0,1\}$ and $b\in \{0,1\}$ to Alice and Bob respectively such
that $a\oplus b = x\wedge y$.  An AND-box can be obtained from any
flavor of oblivious transfer and is defined the same way than an
NL-box\cite{PR94,PR97} without the property that its output can be
obtained before the input of the other player has been provided to the
box (i.e., NL-boxes are non-signaling).  The \emph{equivalence}
between AND-boxes, NL-boxes, and oblivious transfer is discussed in
\cite{WW05b}. At the end of the protocol, each part of the shared key
allowing to decrypt the output must be exchanged in a fair way. For
this task, Alice and Bob rely upon an ideal swap functionality called
\swap.  The result is that any $U$ can be evaluated privately upon any
input provided Alice and Bob have access to one AND-box per
\rgate-gate and one call to the an ideal swap. If the circuit happens
to have only gates in the Clifford group then only one call to an
ideal swap is required for privacy. In other words, \swap\ is
universal for the private evaluation of circuits in the Clifford group
(i.e., those circuits having no \rgate-gate) and itself belongs to
that group (\swap\ is not a classical primitive). To some extent,
circuits in the Clifford group are the {\em easy} ones.  Privacy for
circuits containing \rgate-gates however needs a classical
cryptographic primitive to be evaluated privately by our protocol. It
means that AND-boxes are universal for the private evaluation of any
circuit against specious adversaries. We don't know whether there
exist some unitary transforms that are universal for the private
evaluation of any unitary against specious adversaries.

\paragraph{Previous works.}
All impossibility results in quantum cryptography we are aware of
apply to classical primitives. In fact, the impossibility proofs
usually rely upon the fact that an adversary with a seemingly honest
behavior can force the implementation of classical primitives to
behave quantumly. The result being that implemented that way, the
primitive must leak information to the adversary. This is the spirit
behind the impossibility of implementing oblivious transfer securely
using quantum communication\cite{Lo97}. In that same paper the
impossibility of any one-sided private evaluation of non-trivial
primitives was shown. All these results can be seen as generalizations
of the impossibility of bit commitment schemes based on quantum
communication\cite{LC97,Mayers97}.  The most general impossibility
result we are aware of applies to any {\em non-trivial} two-party
classical function\cite{SSS09}.  It states that it suffices for the
adversary to {\em purify} its actions in order for the quantum
primitive to leak information. An adversary purifying its actions is
specious as defined above. None of these impossibility proofs apply to
quantum primitives characterized by some unitary transform applied to
joint quantum inputs. Blind quantum computation is a primitive
that shows similarities to ours. In \cite{BFK09}, a protocol allowing
a client to get its input to a quantum circuit evaluated blindly has
been proposed. The security of their scheme is unconditional 
while in our setting almost no unitary allows for  unconditional
privacy.

An unpublished work of Smith\cite{Smith06} 
shows how one can devise a private protocol
for the evaluation of any unitary 
that seems to remain private against all quantum adversaries. 
However, the
techniques used require strong cryptographic assumptions like
homomorphic encryption schemes, zero-knowledge and witness 
indistinguishable proof systems. 
The construction is in the 
spirit of protocols for multiparty quantum computation\cite{BCGHS06,CGS02}
and fault tolerant quantum circuits\cite{Shor96,AB97}.
Although our protocol only
guarantees privacy  against specious adversaries, it 
is obtained using much weaker cryptographic assumptions.

\fullv{\paragraph{Organization.}
We introduce protocols for the two-party evaluation of unitaries in 
Sect.~\ref{twoparty}. In Sect.~\ref{semihonest}, we define 
the class of specious
quantum adversaries and  in Sect.~\ref{privacy}, we 
define privacy.
  We show in Sect.~\ref{impossiblesect} that no private
protocol  
 exists for swap. The description of our protocol follows in
Sect.~\ref{prot} and the proof of privacy is in
Appendix~\ref{privacyproof}.}{}

\section{Preliminaries}
\label{sect:preliminaries}

The $N$-dimensional complex Euclidean space (i.e., Hilbert space) will
be denoted by \hdim{N}. We denote quantum registers using calligraphic
typeset \hreg{A}. As usual, $\hreg{A} \otimes \hreg{B}$ denotes the
space of two such quantum registers.  We write $\hreg{A}\approx
\hreg{B}$ when $\hreg{A}$ and $\hreg{B}$ are such that
$\dim{(\hreg{A})}=\dim{(\hreg{B})}$. A register $\hreg{A}$ can undergo
transformations as a function of time; we denote by $\hreg{A}_i$ the
state of space $\hreg{A}$ at time $i$.  When a quantum computation is
viewed as a circuit accepting input in \hreg{A}, we denote all wires
in the circuit by $\wire{w}\in \hreg{A}$.  If the circuit accepts
input in $\hreg{A}\otimes\hreg{B}$ then the set of all wires is
denoted $\wire{w}\in \hreg{A}\cup \hreg{B}$.

The set of all linear mappings from \hreg{A}\ to \hreg{B}\ is denoted
by \lin{\hreg{A}}{\hreg{B}}\ while \linop{\hreg{A}}\ stands for
\lin{\hreg{A}}{\hreg{A}}. To simplify notation, for $\rho\in
\linop{\hreg{A}}$ and $M\in \lin{\hreg{A}}{\hreg{B}}$ we write $M\cdot
\rho$ for $M\rho M^{\dagger}$.
\fullv{

}{}
We denote by $\pos{\hreg{A}}$ the set of positive semi-definite
operators in \hreg{A}. The set of positive semi-definite operators
with trace $1$ acting on \hreg{A}\ is denoted $\dens{\hreg{A}}$;
$\dens{\hreg{A}}$ is the set of all possible quantum states for
register $\reg{A}$.  An operator $A\in\lin{\hreg{A}}{\hreg{B}}$ is
called a {\em linear isometry} if $A^{\dagger}A=\Idt_{\hreg{A}}$.  The
set of unitary operators (i.e., linear isometries with
$\hreg{B}=\hreg{A}$) acting in \hreg{A}\ is denoted by
$\unit{\hreg{A}}$.  The identity operator in \hreg{A}\ is denoted
$\Idt_{\hreg{A}}$ and the completely mixed state in \dens{\hreg{A}}\
is denoted by $\id_{\hreg{A}}$.  For any positive integer $N>0$,
$\Idt_{N}$ and $\id_N$ denote the identity operator respectively the
completely mixed state in \hdim{N}. When the context requires, a pure
state $\ket{\psi}\in \hreg{AB}$ will be written
$\ket{\psi}^{\hreg{AB}}$ to make explicit the registers in which it is
stored.

A linear mapping $\Phi:\linop{\hreg{A}} \mapsto \linop{\hreg{B}}$ is
called a {\em super-operator} since it belongs to
$\lin{\linop{\hreg{A}}}{\linop{\hreg{B}}}$.  $\Phi$ is said to be {\em
positive} if $\Phi(A)\in\pos{\hreg{B}}$ for all $A\in\pos{\hreg{A}}$.
The super-operator $\Phi$ is said to be {\em completely positive} if
$\Phi \otimes \Idt_{\linop{\hreg{Z}}}$ is positive for every choice of
the Hilbert space $\hreg{Z}$.  A super-operator $\Phi$ can be
physically realized or is {\em admissible} if it is completely
positive and preserves the trace: $\trace{\Phi(A)}=\trace{A}$ for all
$A\in \linop{\hreg{A}}$. We call such a super-operator a {\em quantum
operation}. \fullv{Any quantum operation $\Phi:\linop{\hreg{A}} \mapsto
\linop{\hreg{B}}$ can be written in its Kraus form
$\{E_j\}_{j=1}^{\dim{(\hreg{A})}\cdot\dim{(\hreg{B})}}$ where $E_j\in
\lin{\hreg{A}}{\hreg{B}}$ for every $j$ such that $\Phi(\rho) = \sum_j
E_j \rho E_j^{\dagger},$ for any $\rho\in\pos{\hreg{A}}$ and where
$\sum_j E^{\dagger}_j E_j = \Idt_{\hreg{B}}$.}{} Another way to represent
any quantum operation is through a linear isometry $W\in
\lin{\hreg{A}}{\hreg{B}\otimes\hreg{Z}}$ such that $\Phi(\rho) =
\trace[\hreg{Z}]{W\cdot \rho},$ for some extra space
$\hreg{Z}$. Any such isometry $W$ can be implemented by a physical
process as long as the resource to implement\fullv{ the space}{} \hreg{Z}\
is available. This is just a unitary transform in
$\unit{\hreg{A}\otimes\hreg{Z}}$ where the system in \hreg{Z}\ is
initially in known state $\ket{0_{\hreg{Z}}}$.

\fullv{\begin{wrapfigure}{r}{0.42\textwidth}
\vspace{-10pt}
\centering \mbox{\scriptsize
  \Qcircuit @C=1em @R=.7em {
    \lstick{\ket{\psi}} & \multigate{1}{Bell} & \ustick{x} \cw & \cw &
       \control\cw & \\
    \twowire{\ket{\Psi_{0,0}}}& \ghost{Bell}        & \ustick{z} \cw &
       \control\cw & \cwx & \\
       & \qw  & \qw & \gate{Z} \cwx & \gate{X} \cwx &
          \rstick{\ket{\psi}}\qw \gategroup{1}{4}{3}{5}{1em}{--}
  }}
\small{\caption{The teleportation circuit}
\label{fig:teleportation}}
\vspace{-20pt}
\end{wrapfigure}
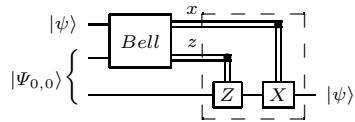}{}

For two states $\rho_0,\rho_1 \in \dens{\hreg{A}}$, we denote by
$\Delta(\rho_0,\rho_1)$ the trace norm distance between $\rho_0$ and
$\rho_1$: $\Delta(\rho_0,\rho_1) := \frac{1}{2} \|
\rho_0-\rho_1\|$. If $\Delta(\rho_0,\rho_1)\leq \varepsilon$ then any
quantum process applied to $\rho_0$ behaves exactly as for $\rho_1$
except with probability at most $\varepsilon$~\cite{RK05}.

\fullv{We let $C_1$ be the Pauli group (the set of tensor products of the
three Pauli matrices $X,$ $Y,$ and $Z$, see Appendix~\ref{commut}, and
the $2 \times 2$ identity matrix $\Idt_2$). Furthermore, $C_{i+1}$ is
then defined recursively for $i\geq 1$ as $C_{i+1} \equiv \{ U | U C_1
U^\dag \in C_{i} \}$, where $C_2$ is called the Clifford group.

}
{ Let $X$,$Y$, and $Z$ be the three non-trivial one-qubit Pauli operators.
}
The Bell measurement is a complete orthogonal measurement on two
qubits made out of the measurement operators\fullv{ $\{\proj{\Psi_{0,0}},
\proj{\Psi_{0,1}}, \proj{\Psi_{1,0}}, \proj{\Psi_{1,1}}\}$ where
$\ket{\Psi_{0,0}}:=
\frac{1}{\sqrt{2}}(\ket{00}+\ket{11})$, 
$\ket{\Psi_{0,1}}:= 
\frac{1}{\sqrt{2}}(\ket{00}-\ket{11})$,
$\ket{\Psi_{1,0}}:= 
\frac{1}{\sqrt{2}}(\ket{01}+\ket{10})$, and 
$\ket{\Psi_{1,1}}:= 
\frac{1}{\sqrt{2}}(\ket{01}-\ket{10})$.
The outcome $\ket{\Psi_{x,z}}$ of the Bell measurement is identified by the 
two classical bits $(x,z) \in \{0,1\}^2$.}{
$\{\proj{\Psi_{x,y}}\}_{x,y\in\{0,1\}}$
where $\ket{\Psi_{x,y}} := \frac{1}{\sqrt{2}}(\ket{0, x} +  (-1)^y \ket{1,\overline{x}})$.
We say that the outcome  of a Bell measurement is $(x,y)\in\{0,1\}^2$ if
$\proj{\Psi_{x,y}}$ has been observed.
}
The quantum one-time-pad is a perfectly secure encryption of quantum
states\cite{AMTW00}. It encrypts a qubit $\ket{\psi}$ as $X^x
Z^z\ket{\psi}$, where the key is two classical bits, $(x,z) \in
\{0,1\}^2$ and $X^0 Z^0 = \Idt$, $X^0 Z^1 = Z$, $X^1 Z^0 = X$ and $X^1
Z^1 = Y$ are the Pauli operators.\fullv{ Quantum
teleportation\cite{BBCJPW93} can be used to implement the quantum
one-time-pad. Consider the teleportation circuit in
Fig.~\ref{fig:teleportation}. If the state to encrypt is $\ket{\psi}$
then the state of the lower wire  before entering the
out-dashed box is the encryption of $\ket{\psi}$ under a uniformly
random key produced by the Bell measurement. The two gates inside the
dashed-box is the decryption circuit.}{}

\newcommand{\rhoin}{\ensuremath{\rho_{\mathrm{in}}}}

\subsection{Modeling two-party strategies}\label{twoparty}

Consider an interactive two-party strategy $\Pi^{\qop{O}}$ between
parties $\aprot$ and $\bprot$ and oracle calls \qop{O}.
$\Pi^{\qop{O}}$ can be modeled by a sequence of quantum operations for
each player together with some oracle calls also modeled by quantum
operations.  Each quantum operation in the sequence corresponds to the
action of one party at a certain step of the strategy.  The following
definition is a straightforward adaptation of $n$-turn interactive
quantum strategies as described in \cite{GW07}.  The main difference
is that here, we provide a joint input state to both parties and that
quantum transmissions taking place during the execution is modeled by
a quantum operation; one that is moving a state on one party's side to the
other party.  

\begin{definition}
A {\em $n$--step two party strategy with oracle calls} denoted
$\Pi^{\qop{O}}=(\aprot,\bprot,\qop{O},n)$ consists of:
\begin{enumerate}
\item 

  input spaces ${\hreg{A}}_0$ and ${\hreg{B}}_0$ for parties \aprot\ and
  \bprot\ respectively,

\item 

  memory spaces ${\hreg{A}}_1,\ldots,{\hreg{A}}_n$ and
  ${\hreg{B}}_1,\ldots,{\hreg{B}}_n$ for \aprot\ and \bprot\
  respectively,

\item 

  an $n$-tuple of quantum operations $(\alice{1},\ldots,\alice{n})$
  for \aprot, $\alice{i}: \linop{{\hreg{A}}_{i-1}} \mapsto
  \linop{{\hreg{A}}_{i}}$, $(1\leq i \leq n)$,

\item 

  an $n$-tuple of quantum operations $(\bob{1},\ldots,\bob{n})$ for
  \bprot, $\bob{i}: \linop{{\hreg{B}}_{i-1}} \mapsto
  \linop{{\hreg{B}}_{i}}$, $(1\leq i \leq n)$,

\item 

  memory spaces $\hreg{A}_1,\ldots, \hreg{A}_n$ and
  $\hreg{B}_1,\ldots,\hreg{B}_n$ can be written as $\hreg{A}_i
  =\ao{i}\otimes \hreg{A}'_i$ and $\hreg{B}_i = \bo{i}\otimes
  \hreg{B}'_i$, $(1\leq i \leq n)$, and $\qop{O} =
  (\qop{O}_1,\qop{O}_2,\ldots, \qop{O}_n)$ is an $n$-tuple of quantum
  operations: $\qop{O}_i : \linop{\ao{i}\otimes\bo{i}} \mapsto
  \linop{\ao{i}\otimes\bo{i}}$, $(1\leq i \leq n)$.
   
\end{enumerate}
If $\Pi=(\aprot,\bprot,n)$ is a $n$-turn two-party protocol then the
final state of the interaction upon input state $\rhoin\in
\dens{{\hreg{A}}_0\otimes {\hreg{B}}_0 \otimes \hreg{R}}$, where
$\hreg{R}$ is a system of dimension $\dim \hreg{R} = \dim \hreg{A}_0
\dim \hreg{B}_0$, is:
\[ 
\begin{split}
[\aprot\compose\bprot](\rhoin) := & (\Idt_{\linop{\hreg{A}'_n \otimes \hreg{B}'_n\otimes \hreg{R}}}\otimes\qop{O}_n)(\alice{n}\otimes\bob{n} \otimes  \Idt_\hreg{R})\\
&\ldots(\Idt_{\linop{\hreg{A}'_1 \otimes \hreg{B}'_1\otimes \hreg{R}}}\otimes\qop{O}_1)(\alice{1}\otimes\bob{1} \otimes  \Idt_\hreg{R})(\rhoin) \enspace.
\end{split}
\]
\end{definition}
Step $i$ of the strategy corresponds to the actions of $\alice{i}$ and
$\bob{i}$ followed by the oracle call $\qop{O}_i$. 


Note that we consider input states defined on the input systems
together with a reference system $\hreg{R}$; this allows us to show
the correctness and privacy of the protocol not only for pure inputs,
but also for inputs that are entangled with a third party. This is the
most general case allowed by quantum mechanics.

A two-party strategy is therefore defined by quantum operation tuples
$(\alice{1},\ldots,\alice{n})$, $(\bob{1},\ldots,\bob{n})$, and
$(\qop{O}_1,\ldots,\qop{O}_n)$. These operations also
define working spaces ${\hreg{A}}_0,\ldots,{\hreg{A}}_n, 
{\hreg{B}}_0,\ldots,{\hreg{B}}_n$ together with the input-output spaces to the oracle calls
 $\hreg{A}_i^{\qop{O}}$ and $\hreg{B}_i^{\qop{O}}$ for $1\leq i \leq n$.
%

A \emph{communication oracle} from Alice to Bob is modeled by having
$\ao{i} \approx \bo{i}$ and letting $\qop{O}_i$ move the state in
$\ao{i}$ to $\bo{i}$ and erase $\ao{i}$. Similarly for
communication in the other direction. We define a \emph{bare model}
protocol to be one which only uses communication oracles.

\section{Specious Quantum Adversaries}\label{semihonest}

\subsection{Protocols for  two-party evaluation}\label{protdef}

Let us consider two-party protocols for the quantum evaluation of
unitary transform $U\in \unit{{\hreg{A}}_0\otimes{\hreg{B}}_0}$
between parties \aprot\ and \bprot\ upon joint input state
$\rhoin\in\dens{{\hreg{A}}_0\otimes{\hreg{B}}_0 \otimes \hreg{R}}$\fullv{.
We define these protocols as two-party interactive strategies with
placeholder for the output as follows:}{:}

\begin{definition}\label{twopartyprot}
  A \emph{two-party protocol $\puq{U}=(\aprot,\bprot,\qop{O},n)$ for
  $U\in \unit{{\hreg{A}}_0\otimes{\hreg{B}}_0}$} is an $n$--step
  two-party strategy with oracle calls, where $\hreg{A}_{n}\approx
  \hreg{A}_0$ and $\hreg{B}_n \approx \hreg{B}_0$. It is said to be
  \emph{$\varepsilon$--correct} if
  \[ \Delta\left([\aprot\compose\bprot](\rhoin), 
     (U \otimes \Idt_{\hreg{R}})\cdot\rhoin\right) 
     \leq \varepsilon\,\,\,\,\,\, \text{ for all } \rhoin\in
     \dens{{\hreg{A}}_0\otimes{\hreg{B}}_0 \otimes \hreg{R}}\enspace .
     \]
  We denote by $\pu{U}$ a two-party protocol in the {\em bare model}
  where, without loss of generality, we assume that $\qop{O}_{2i+1}$
  ($0\leq i \leq \lfloor \frac{n}{2} \rfloor$) implements a
  communication channel from \aprot\ to \bprot\ and $\qop{O}_{2i}$
  ($1\leq i \leq \lfloor \frac{n}{2} \rfloor$) implements a
  communication channel from \bprot\ to \aprot. Communication oracles
  are said to be {\em trivial}.
\end{definition}
In other words, a two-party protocol $\puq{U}$ for unitary $U$ is a
two-party interactive strategy where, at the end, the output of the
computation is stored in the memory of the players. $\puq{U}$ is
correct if, when restricted to the output registers (and $\hreg{R}$),
the final quantum state shared by \aprot\ and \bprot\ is $(U \otimes
\Idt_{\hreg{R}})\cdot \rhoin$.

As it will become clear when we discuss privacy in
Sect.~\ref{privacy}, we need to consider the joint state at any step
during the evolution of the protocol\fullv{.  We define,}{:}
\begin{align}
  \hstatei{1}{\rhoin} & := (\Idt_{\linop{\hreg{A}'_1\otimes
      \hreg{B}'_1\otimes\hreg{R}}}\otimes \qop{O}_1) (\alice{1}\otimes
      \bob{1}\otimes \Idt_{\linop{\hreg{R}}})(\rhoin), \nonumber \\
  \hstatei{i+1}{\rhoin} &:= (\Idt_{\linop{\hreg{B}'_{i+1} \otimes
       \hreg{A}'_{i+1}\otimes\hreg{R}}}\otimes\qop{O}_{i+1})(\alice{i+1}\otimes
       \bob{i+1} \otimes\Idt_{\linop{\hreg{R}}}) (\hstatei{i}{\rhoin})
       \enspace , \label{statei}
\end{align}
for $1\leq i <n$.
We also write the final state of $\puq{U}$  upon input state $\rhoin$ as
$\hstatei{n}{\rhoin} = [\aprot\compose\bprot](\rhoin)$.

\subsection{Modeling Specious Adversaries}

\newcommand{\maxta}{\ensuremath{t_\aprot}}
\newcommand{\maxtb}{\ensuremath{t_\bprot}}

Intuitively, a specious adversary acts in any way apparently
indistinguishable from the honest behavior, in the sense that 
no audit can distinguish the behavior of the adversary from
the honest one. 

More formally, a specious adversary in
$\puq{U}=(\aprot,\bprot,\qop{O},n)$ may use an arbitrary large quantum
memory space. However, at any step $1\leq i\leq n$, the adversary can
transform its own current state to one that is indistinguishable from
the honest joint state. These transforms are modeled by quantum
operations, one for each step of the adversary in $\puq{U}$, and are
part of the adversary's specification. We denote by
$(\trans{1},\ldots,\trans{n})$ these quantum operations where
$\trans{i}$ produces a valid transcript at the end of the $i$--th
step.
    
Let $\ada$ and $\adb$ be adversaries in $\puq{U}$.  We denote by
$\puq{U}(\ada) = (\ada,\bprot,\qop{O},n)$ and
$\puq{U}(\adb)=(\aprot,\adb,\qop{O},n)$ the resulting $n$--step
two-party strategies.  We denote by $\astatei{i}{\ada}{\rhoin}$ the
state defined in (\ref{statei}) for protocol $\puq{U}(\ada)$ and
similarly by $\astatei{i}{\adb}{\rhoin}$ that state for protocol
$\puq{U}(\adb)$.

Adding the possibility for the adversary to be
$\varepsilon$-\emph{close} to honest, we get the following definition:

\begin{definition}\label{specious}
  Let $\puq{U} = (\aprot,\bprot,\qop{O},n)$ be an $n$--step two-party
  protocol with oracle calls for $U \in \unit{{\hreg{A}}_0\otimes
  {\hreg{B}}_0}$. We say that:
  \begin{itemize}
  \item 

    \ada\ is \emph{$\varepsilon$--specious} if $\puq{U}(\ada) =
    (\ada,\bprot,\qop{O},n)$ is an $n$--step two-party strategy with
    $\tilde{\hreg{A}}_0={\hreg{A}}_0$ and there exists a sequence of
    quantum operations $(\trans{1},\ldots,\trans{n})$ such that:
    \begin{enumerate}
    \item 

      for every $1\leq i \leq n$, $\trans{i}:\linop{\tilde{\hreg{A}}_i}
      \mapsto \linop{{\hreg{A}}_i}$,

    \item 

      for every input state $\rhoin \in \dens{{\hreg{A}}_0\otimes
	{\hreg{B}}_0 \otimes \hreg{R}}$, and for all $1\leq i \leq n$,
      \[ \Delta\left((\trans{i}\otimes \Idt_{\linop{\hreg{B}_i \otimes
	\hreg{R}}})
      \left(\astatei{i}{\ada}{\rhoin}\right),
      \hstatei{i}{\rhoin}\right) \leq \varepsilon\enspace .
      \]
    \end{enumerate}

  \item 

    \adb\ is \emph{$\varepsilon$--specious} if $\puq{U}(\adb) =
    (\aprot,\adb,\qop{O},n)$ is a $n$--step two-party strategy with
    $\tilde{\hreg{B}}_0={\hreg{B}}_0$ and there exists a sequence of
    quantum operations $(\trans{1},\ldots,\trans{n})$ \fullv{such that:
    \begin{enumerate}
    \item 

      for every $1\leq i \leq n$,
      $\trans{i}:\linop{\tilde{\hreg{B}}_{i}} \mapsto
      \linop{{\hreg{B}}_{i}}$,

    \item  
      
      for every input state $\rhoin \in \dens{{\hreg{A}}_0\otimes
	{\hreg{B}}_0 \otimes \hreg{R}}$, and for all $1\leq i \leq n$,
      \[ \Delta\left((\Idt_{\linop{\hreg{A}_i \otimes
	  \hreg{R}}}\otimes\trans{i})
          \left(\astatei{i}{\adb}{\rhoin}\right),
          \hstatei{i}{\rhoin}\right) \leq \varepsilon\enspace .
      \]
    \end{enumerate}}{defined as before with $\hreg{B}_i, \tilde{\hreg{B}}_{i}$, and $\astatei{i}{\adb}{\rhoin}$
    replacing $\hreg{A}_i, \tilde{\hreg{A}}_{i}$, and  $\astatei{i}{\ada}{\rhoin}$ respectively.}
  \end{itemize}
  If a party is $\varepsilon(m)$--specious with $\varepsilon(m)$
  negligible for $m$ a security parameter then we say that this party
  is {\em statistically specious}.
\end{definition}

\subsection{Privacy}\label{privacy}

Privacy for $\puq{U}$ is defined as the ability for a simulator,
having only access to the adversary's input and the ideal
functionality $U$, to reproduce the state of the adversary at any step
in the execution of $\puq{U}$. Our definition is similar to the one
introduced in \cite{Watrous02} for statistical zero-knowledge proof
systems.

A simulator for an adversary in $\puq{U}$ is represented by a sequence
of quantum operations $( \simul{i} )_{i=1}^n$, where $\simul{i}$
re-produces the view of the adversary after step $i$.  \simul{i}
initially receives the adversary's input and has access to the ideal
functionality for $U$ evaluated upon the joint input of the adversary
and the honest player. Because of no-cloning, a simulator calling $U$
loses its input, and the input might be required to simulate
e.g.~early steps in the protocol, so we have to allow that $\simul{i}$
does not call $U$. For this purpose we introduce a bit
$q_i\in\{0,1\}$. When $q_i=0$, \simul{i}\ does not call $U$ and when
$q_i=1$, $\simul{i}$ must first call the ideal functionality $U$
before performing some post-processing. More precisely,

\begin{definition}\label{sim}
  Let $\puq{U}=(\aprot,\bprot,\qop{O},n)$ be an $n$--step two-party
  protocol for $U\in \dens{\hreg{A}_0\otimes\hreg{B}_0}$. Then,
  \begin{itemize}
  \item 

    \emph{$\simulator{\ada}=\langle (\simul{1},\ldots, \simul{n}), q
      \rangle$ is a simulator for adversary \ada\ in \puq{U}} if it
      consists of:
      \begin{enumerate}
      \item 

	a sequence of quantum operations
	$(\simul{1},\ldots,\simul{n})$ where for $1\leq i \leq n$,
	$\simul{i} : \linop{\hreg{A}_0} \mapsto
	\linop{\tilde{\hreg{A}}_i}$,\label{cc1}

      \item  

	a sequence of bits $q\in \{0,1\}^n$ determining if the simulator calls the ideal functionality at step $i$: $q_i = 1$ iff the simulator calls the ideal functionality.
        \label{cc2}
      \end{enumerate} 

    \item  

      Similarly, \emph{$\simulator{\adb}=\langle (\simul{1},\ldots,
	\simul{n}), q' \rangle$ is a simulator for adversary \adb\ in
	\puq{U}} \fullv{if it consists of:
      \begin{enumerate}
      \item 
	
	a sequence of quantum operations
        $(\simul{1},\ldots,\simul{n})$ where for $1\leq i\leq n$,
        $\simul{i} : \linop{\hreg{B}_0} \mapsto
        \linop{\tilde{\hreg{B}}_{i}}$

      \item 

	a sequence of bits $q'\in \{0,1\}^{n}$ determining if the simulator calls the ideal functionality at step $i$: $q'_i = 1$ iff the simulator calls the ideal functionality.

      \end{enumerate}}{if it satisfies conditions \ref{cc1} and \ref{cc2} above with $q', \hreg{B}_0, \hreg{B}_i$, and $\tilde{\hreg{B}}_i$
      replacing $q,\hreg{A}_0, \hreg{A}_i$, and $\tilde{\hreg{A}}_i$ respectively. }
  \end{itemize}
\end{definition}
Given an input state $\rhoin \in \dens{\hreg{A}_0 \otimes \hreg{B}_0
\otimes \hreg{R}}$, we define the \ada's respectively \adb's simulated views  as:
\begin{align*}
	  \view{\ada}{\rhoin}{i}&:=\trace[\hreg{B}_0]{(\simul{i}\otimes
	  \Idt_{\linop{\hreg{B}_0 \otimes \hreg{R}}})\left( (U^{q_i}
	  \otimes \Idt_{\hreg{R}})\cdot \rhoin\right)}\ ,\\
	  \view{\adb}{\rhoin}{i}&:=\trace[\hreg{A}_0]{(\Idt_{\linop{\hreg{A}_0
	  \otimes \hreg{R}}}\otimes \simul{i})\left( (U^{q'_i} \otimes
	  \Idt_{\hreg{R}})\cdot \rhoin\right)}\enspace .
\end{align*}
We say that protocol $\puq{U}$ is private against specious adversaries
if there exits a simulator for the view at any step of any such
adversary. In more details,
 
\begin{definition}\label{privdef}
  Let $\puq{U}=(\aprot,\bprot,\qop{O},n)$ be a protocol for $U\in
  \unit{\hreg{A}_0\otimes {\hreg{B}}_0}$ and let $0 \leq \delta\leq
  1$.  We say that $\puq{U}$ is {\em $\delta$--private against
  $\varepsilon$--specious \ada} if there exists a simulator
  $\simulator{\ada}$ such that for all input states $\rhoin \in
  \dens{\hreg{A}_0\otimes \hreg{B}_0 \otimes \hreg{R}}$ and for all $1
  \leq i \leq n$, $\Delta\left(\view{\ada}{\rhoin}{i},
  \trace[\hreg{B}_i]{\astatei{i}{\ada}{\rhoin}} \right) \leq \delta.$
  Similarly, we say that $\pu{U}$ is {\em $\delta$--private against
  $\varepsilon$--specious \adb} if there exists a simulator
  $\simulator{\adb}$ such that for all input states $\rhoin \in
  \dens{\hreg{A}_0\otimes \hreg{B}_0 \otimes \hreg{R}}$ and for all $1
  \leq i \leq n$, $\Delta\left(\view{\adb}{\rhoin}{i},
  \trace[\hreg{A}_{i}]{\astatei{i}{\adb}{\rhoin}} \right) \leq
  \delta.$ Protocol $\puq{U}$ is {\em $\delta$--private against
  $\varepsilon$--specious adversaries} if it is $\delta$--private
  against both\delete{$\varepsilon$--specious} \ada\ and \adb. For
  $\gamma>0$, if $\puq{U}$ is $2^{-\gamma
  m}$--private\delete{against specious adversaries} for $m\in \N^+$ a
  security parameter then we say that $\puq{U}$ is \emph{statistically
  private\delete{ against specious adversaries}}.
\end{definition}

\fullv{One should keep in mind that $\delta$ should be kept small compared to the number of rounds, since the protocol is only secure if we can ensure that, with high probability, the adversary cannot behave differently in the simulated world at \emph{any} of the rounds. If $\delta n$ is kept small, we can use the union bound over all the rounds to ensure this.}{}


We show next that for some unitary, statistical privacy cannot be
satisfied by any protocol 
 in the bare model. 
 
\section{Unitaries with no private protocols}\label{impossiblesect}

In this section, we show that no statistically private protocol for
the swap gate exists in the bare model.
The swap gate, denoted \swap,  is the following unitary transform:
\[ \swap : \ket{\phi_A}^{\hreg{A}_0} \ket{\phi_B}^{\hreg{B}_0} \mapsto
   \ket{\phi_B}^{\hreg{A}_0}\ket{\phi_A}^{\hreg{B}_0}\enspace ,
\]
for any one qubit states $\ket{\phi_A}\in \hreg{A}_0$ and
$\ket{\phi_B}\in \hreg{B}_0$ (i.e.,
$\dim{(\hreg{A}_0)}=\dim{(\hreg{B}_0)}=2$).  Notice that \swap\ is in
the Clifford group since it can be implemented with three \cnot\
gates.  It means that universality is not required (gates in the
Clifford groups are not universal for quantum computation) for a
unitary to be impossible to evaluate privately. The impossibility of
\swap\ essentially follows from no cloning.

\begin{theorem}[Impossibility of swapping] \label{noswap}
  There is no correct and statistically private two-party protocol
  $\pu{\swap}=(\aprot,\bprot,\qop{O},n(m))$ in the bare model.
\end{theorem}

\fullv{}{
\begin{proof}
Suppose that there exists an $\varepsilon$-correct, $\varepsilon$-private protocol in the bare model for $\swap$ for sufficiently small $\varepsilon$; we will show that this implies that one of the two players must \emph{lose} information upon receiving a message, which is clearly impossible.

We will consider the following particular pure input state: $\ket{\varphi} := \ket{\Psi_{0,0}}^{\hreg{A}_0 \hreg{R}_{\hreg{A}}} \otimes \ket{\Psi_{0,0}}^{\hreg{B}_0 \hreg{R}_{\hreg{B}}}$, a maximally entangled state between $\hreg{A}_0 \otimes \hreg{B}_0$ and the reference system $\hreg{R}_{\hreg{A}} \otimes \hreg{R}_{\hreg{B}}$ that is broken down into two subsystems for convenience. Furthermore, we will consider the ``purified'' versions of the honest players for this protocol; in other words, we will assume that the super-operators $\aprot_1,\dots,\aprot_n$ and $\bprot_1,\dots,\bprot_n$ are in fact linear isometries and that therefore the players never discard any information unless they have to send it to the other party. The global state $\hstatei{i}{\varphi}$ after step $i$ is therefore a pure state on $\hreg{A}_i \otimes \hreg{B}_i \otimes \hreg{R}_{\hreg{A}} \otimes \hreg{R}_{\hreg{B}}$. 

After step $i$ of the protocol (i.e., after the $i$th message has been sent), Alice's state must either depend only on her own original input (if $q_i = 0$ for her simulator), or on Bob's original input (if $q_i=1$). More precisely, by the definition of privacy (Definition \ref{privdef}), we have that
\[
\Delta\left( \view{\aprot}{\varphi}{i}, \tr_{\hreg{B}_i}[\hstatei{i}{\varphi}] \right) \leq \varepsilon \enspace ,
\]
where $\view{\aprot}{\varphi}{i}$ is \aprot's simulated view after step $i$ and $\hstatei{i}{\varphi}$ is the global state in the real protocol after step $i$. Now, suppose that $q_i=0$, and let $\ket{\xi} \in \hreg{A}_i \otimes \hreg{R}_{\hreg{A}} \otimes \hreg{R}_{\hreg{B}}' \otimes \hreg{Z}$ be a purification of $\view{\aprot}{\varphi}{i}$ with $\hreg{Z}$ being the purifying system, and $\hreg{R}_{\hreg{B}}$ renamed for upcoming technical reasons. The pure state $\ket{\xi} \otimes \ket{\Psi_{0,0}}^{\hreg{R}_{\hreg{B}} \hreg{B}_0}$ has the same reduced density matrix as $\view{\aprot}{\varphi}{i}$ on $\hreg{A}_i \otimes \hreg{R}_{\hreg{A}} \otimes \hreg{R}_{\hreg{B}}$. Hence, by Uhlmann's theorem, there exists a linear isometry $V : \hreg{B}_i \rightarrow \hreg{B}_0 \otimes \hreg{Z} \otimes \hreg{R}_{\hreg{B}}'$ such that
\[
V \view{\aprot}{\varphi}{i} V^\dagger  = \ket{\xi}\bra{\xi} \otimes \ket{\Psi_{0,0}}\bra{\Psi_{0,0}}^{\hreg{B}_0\hreg{R}_{\hreg{B}}}
\]
and hence
\[
\Delta\left( V \hstatei{i}{\varphi} V^\dagger, \ket{\xi}\bra{\xi} \otimes \ket{\Psi_{0,0}}\bra{\Psi_{0,0}}^{\hreg{B}_0\hreg{R}_{\hreg{B}}} \right) \leq \sqrt{2\varepsilon}
\enspace .
\]
This means that if $q_i = 0$, then Bob is still capable of reconstructing his own input state after step $i$ by applying $V$ to his working register. Clearly, this means that $q_i' = 0$ (i.e., Bob's simulator must also not call $\swap$), and therefore, by the same argument, Alice must also be able to reconstruct her own input with an isometry $V_A : \hreg{A}_i \rightarrow \hreg{B}_0 \otimes \hreg{Z} \otimes \hreg{R}'_{\hreg{A}}$. The same argument also holds if $q_i = 1$: we then conclude that $q_i' = 1$ and that Alice and Bob must have each other's inputs; no intermediate situation is possible. We conclude that, at every step $i$ of the protocol, $q_i = q'_i$.

Now, before the protocol starts, Alice must have her input, and Bob must have his, hence, $q_0 = q_0' = 0$. At the end, the two inputs must have been swapped, which means that $q_n = q'_n = 1$; there must therefore be a step $k$ in the protocol after which the two inputs are swapped but not before, meaning that $q_k = 1$ and $q_{k-1} = 0$. But at each step, only one player receives information, which means that at this step $k$, the player who received the message must lose the ability to reconstruct his own input, which is clearly impossible. \qed
\end{proof}
}

Using this line of reasoning, Theorem~\ref{noswap} can be extended to apply to any protocol for almost any unitary preventing both parties to recover their input states from its output.


\subsubsection{Sufficient Assumptions for Private \swap.}
A private protocol for \swap\ in the bare model would exist if the
players could rely on special relativity and a lower bound on
their separation in space: they simply send their messages
simultaneously. The fact that messages cannot travel faster than the
speed of light ensures that the messages are independent of each
other.
It is also straightforward to devise a private protocol for \swap\
based on commitment schemes.  \aprot\ sends one half EPR-pair to
\bprot\ while keeping the other half. \aprot\ then teleports (without
announcing the outcome of the measurement) her register and commits on
the outcome of the Bell measurement. \bprot\ sends his register to
\aprot\ before she opens her commitment. \fullv{This allows \bprot\ to
reconstruct \aprot's initial state.}{}

\section{The Protocol}\label{prot}

\newcommand{\qcol}[1]{\ensuremath{\mathbf{Co}_{#1}}}
\newcommand{\qcnl}{\ensuremath{\mathbf{AND}}}
\newcommand{\ikru}{\ensuremath{n_U}}

We now describe a private protocol for the two-party evaluation of any
unitary $U\in \unit{\hreg{A}_0\otimes\hreg{B}_0}$ denoted by
$\pup{U}^{\qop{O}}=(\aprots,\bprots,\qop{O},n_U+1)$ where $U$ is represented by
a circuit $C_U$ with $u$ gates in \ug.  
We slightly abuse the notation with respect to the parameter $n_U+1$.
Given circuit $C_U$, we let $n_U$ be the number of oracle
calls (including calls to communication oracles). Setting the last 
parameter to $n_U+1$ instead of $n_U$ comes from the fact
that in our protocol, \aprots\ and \bprots\ have to perform a last
operation each in order to get their outcome. These last operations
do not involve a call to any oracle.   
Let $G_j$ be the $j$-th gate
in $C_U = G_{u} G_{u-1} \ldots G_1$.  The protocol is obtained by
composing sub-protocols for each gate similarly to well-known
classical constructions\cite{Yao86,kilian88}.  Notice that
$\pup{U}^{\qop{O}}$ will not be presented in the form of
Definition~\ref{twopartyprot}.  \aprots\ is not necessarily sending
the first and the last messages. This can be done without consequences
since we provide a simulation for each step where a message from the
honest party is received or the output of a call to an ideal
functionality is available. Putting $\pup{U}^{\qop{O}}$ in the
standard form of Definition~\ref{twopartyprot} is straightforward and
changes nothing to the proof of privacy.

The evaluation of each gate is performed over shared encrypted states.
Each wire in $C_U$ will be updated from initially holding the input
$\rhoin \in \dens{\hreg{A}_0\otimes\hreg{B}_0\otimes\hreg{R}}$ to
finally holding the output $(U\otimes\Idt_{\hreg{R}} )\cdot \rhoin \in
\dens{\hreg{A}_0\otimes \hreg{B}_0\otimes\hreg{R}}$.  The state of
wires $\wire{w}\in \hreg{A}_0\cup\hreg{B}_0$ after the evaluation of
$G_j$ are stored at \aprots's or \bprots's according if $\wire{w}\in
\hreg{A}_0$ or $\wire{w}\in \hreg{B}_0$.  The shared encryption keys
for wire $\wire{w}\in \hreg{A}_0 \cup \hreg{B}_0$ updated after the
evaluation of $G_j$ are denoted by
$\keya{w}{j}=(\kxa{w}{j},\kza{w}{j})\in \{0,1\}^2$ and
$\keyb{w}{j}=(\kxb{w}{j},\kzb{w}{j})\in \{0,1\}^2$ for
\aprots\ and \bprots\ respectively and are held privately in internal
registers of each party.

The final phase of the protocol is where a call to an ideal
functionality is required.  \aprots\ and \bprots\ exchange their own
part of each encryption key for the other party's wires. In order
to do this, the {\em key-releasing phase} invokes an ideal \swap-gate
as functionality:
$\qop{O}_{\ikru}:\linop{\hreg{A}^{\qop{O}}_{\ikru}\otimes
   \hreg{B}^{\qop{O}}_{\ikru}} \mapsto
   \linop{\hreg{A}^{\qop{O}}_{\ikru}\otimes
   \hreg{B}^{\qop{O}}_{\ikru}}$,
where $\qop{O}_{\ikru}(\rho):= \swap\cdot \rho$.  Upon
joint input state $\rhoin\in
\dens{\hreg{A}_0\otimes\hreg{B}_0\otimes\hreg{R}}$, protocol
$\pup{U}^{\qop{O}(U)}$ runs the following phases:
\begin{description}
\item[Initialization:] 

  We assume that \aprots\ and \bprots\ have agreed upon a description
  of $U$ by a circuit $C_U$ made out of $u$ gates $(G_1,\ldots,
  G_{u})$ in \ug.  For all wires $\wire{w}\in\hreg{A}_0\cup
  \hreg{B}_0$, \aprots\ and $\bprots$ set their initial encryption
  keys as $\keya{w}{0}=(\kxa{w}{0},\kza{w}{0}):=(0,0)$ and
  $\keyb{w}{0}=(\kxb{w}{0},\kzb{w}{0}):=(0,0)$ respectively.

\item[Evaluation:] 

  For each gate number $1\leq j \leq u$, \aprots\ and \bprots\
  evaluate $G_j$ as described in details below.  This evaluation
  results in shared encryption under keys
  $\keya{w}{j}=(\kxa{w}{j},\kza{w}{j})$ and
  $\keyb{w}{j}=(\kxb{w}{j},\kzb{w}{j})$ for all wires $\wire{w}\in
  \hreg{A}_0\cup \hreg{B}_0$, which at that point hold a shared
  encryption of $((G_j G_{j-1} \ldots
  G_{1})\otimes\Idt_{\hreg{R}})\cdot \rhoin$. Only the evaluation of
  the \rgate-gate requires a call to an ideal functionality (i.e., an
  \nl).

\item[Key-Releasing:]  

  Let $\hreg{A}^{\qop{O}}_{\ikru}$ and $\hreg{B}^{\qop{O}}_{\ikru}$ be
  the set of registers holding respectively
  $\keya{w}{u}=(\kxa{w}{u},\kza{w}{u})$ for $\wire{w}\in \hreg{B}_0$
  and $\keyb{w}{u}=(\kxb{w}{u},\kzb{w}{u})$ for $\wire{w}\in
  \hreg{A}_0$. We assume w.l.g that dimensions of both sets of
  registers are identical\footnote{Otherwise, add enough registers
  initially in state $\ket{0}$ to the smaller set.}:
  \begin{enumerate}
  \item 

    \aprots\ and \bprots\ run the ideal functionality for the
    \swap-gate upon registers $\hreg{A}^{\qop{O}}_{\ikru}$ and
    $\hreg{B}^{\qop{O}}_{\ikru}$.

  \item 

    \aprots\ applies the decryption operator
    $\keya{w}{}=(\kxa{w}{u}\oplus \kxb{w}{u}, \kza{w}{u}\oplus
    \kzb{w}{u})$ to each of her wires $\wire{w}\in \hreg{A}_0$.

  \item 

    \bprots\ applies the decryption operator for key
    $\keyb{w}{}=(\kxa{w}{u}\oplus \kxb{w}{u}, \\ \kza{w}{u}\oplus
    \kzb{w}{u})$ to each of his wires $\wire{w}\in \hreg{B}_0$.

  \end{enumerate}
\end{description}
\fullv{In the following subsections \ref{sect:crypt} to
\ref{sect:rgate}, we describe the evaluation phase 
for each gate in $\ug$.}{}

\subsubsection{Swapping for key-releasing.}
Notice that the key-releasing phase only uses the \swap-gate with classical input states. The reader might therefore wonder why this functionality is defined quantumly when a classical swap would work equally well.  The reason is that, perhaps somewhat surprisingly, a classical swap is a potentially stronger primitive than a quantum swap. From a classical swap one can build a quantum swap by encrypting the quantum states with classical keys, exchange the encrypted states using quantum communication, and then using the classical swap to exchange the keys. Obtaining a classical swap from a quantum one, however, is not obvious.  Suppose that registers \hreg{A}\ and \hreg{B}\ should be swapped classically while holding quantum states beforehand. These registers could be entangled with some purification registers before being swapped.  Using a quantum swap between \hreg{A}\ and \hreg{B}\ will always leave these registers entangled with the purification registers until they become measured while a classical swap will ensure that \hreg{A}\ and \hreg{B}\ become unentangled with the purification registers after its invocation. In other words, a classical swap could prevent an adversary from exploiting entanglement in his attack.

\subsubsection{The ideal AND-box functionality.} 

As we are going to see next, a call to an ideal AND-box is required
during the evaluation of the \rgate-gate. Unlike the ideal \swap\ used
for key-releasing, the AND-box will be modeled by a purely classical
primitive denoted \nl.  This is required for privacy of our protocol
since any implementation of it by some unitary will necessarily
leak\cite{SSS09}.  The quantum operation implementing it will first
measure the two one-qubit input registers in the computational basis
in order to get classical inputs $x,y \in\{0,1\}$ for \aprots\ and
\bprots\ respectively. The classical output bits are then set to
$a\in_R\{0,1\}$ for \aprots\ and $b=a\oplus x y$ for \bprots.

\subsection{Computing over Encrypted States}\label{sect:crypt}

Before the execution of $G_{j+1}$ in $C_U$, \aprots\ and \bprots\
share an encryption of $\rho_j = \left((G_j\cdot G_{j-1} \cdot \ldots
\cdot G_1)\otimes \Idt_{\hreg{R}}\right) \cdot \rhoin$ in registers\footnote{To ease the notation in
the following, we assume $\rho_j\in \dens{\hreg{A}_0 \otimes
\hreg{B}_0}$ rather than in $\dens{\hreg{A}_0 \otimes
\hreg{B}_0\otimes\hreg{R}}$. It is easy to see that this can be done
without loss of generality.}
holding wires $\wire{w}\in \hreg{A}_0\cup \hreg{B}_0$.  Each wire
$\wire{w}\in \hreg{A}_0 \cup \hreg{B}_0$ is encrypted by a shared
quantum one-time pad as
\begin{equation}\label{encrypt}
  \left(
  \left(
  \bigotimes_{\wire{w}\in \hreg{A}_0 \cup \hreg{B}_0}
  X^{X_{\aprots}^j(\wire{w})
    \oplus X_{\bprots}^j(\wire{w})}Z^{Z_{\aprots}^j(\wire{w})\oplus
  Z_{\bprots}^j(\wire{w})}\right)\otimes \Idt_{\hreg{R}} \right) \cdot
  \rho_j\enspace ,
\end{equation}
where $K_{\aprots}^{j}(\wire{w}) :=
(X_{\aprots}^j(\wire{w}),Z_{\aprots}^j(\wire{w}))\in \{0,1\}^2$ and
$K_{\bprots}^j(\wire{w}) :=
(X_{\bprots}^j(\wire{w}),Z_{\bprots}^j(\wire{w}))\in \{0,1\}^2$ are
two bits of secret keys for \aprots\ and \bprots\ respectively.  In
other words, wires $\wire{w}\in \hreg{A}_0 \cup \hreg{B}_0$ are
encrypted by $X^{x}Z^{z}$ where $x=X_{\aprots}^j(\wire{w}) \oplus
X_{\bprots}^j(\wire{w})$ and $z=Z_{\aprots}^j(\wire{w})\oplus
Z_{\bprots}^j(\wire{w})$ are additive sharings for the encryption of
$\wire{w}$.  Then, evaluating $G_{j+1}$ upon state (\ref{encrypt})
will produce a new sharing
$K_A^{j+1}(\wire{w}):=(X_A^{j+1}(\wire{w}),Z_A^{j+1}(\wire{w}))$ and
$K_B^{j+1}(\wire{w}):= (X_B^{j+1}(\wire{w}),Z_B^{j+1}(\wire{w}))$ for
the encryption of state $\rho_{j+1} =
(G_{j+1}\otimes\Idt_{\hreg{R}})\cdot \rho_j $.  In the following, we
describe how to update the keys for the wires involved in the current
gate to be evaluated---all other wires retain their previous
values.

\subsection{Evaluation of Gates in the Pauli and Clifford Groups}\label{cliffeval}

\fullv{\subsubsection{Pauli gates.}}{}
Non-trivial Pauli gates (i.e., $X,Y,$ and $Z$) can easily be computed
on encrypted quantum states since they commute or anti-commute pairwise.
Let $G_{j+1}\in \{X,Y,Z\}$ be the Pauli gate to be executed on wire
$\wire{w}$. 
\fullv{We have: 
\begin{equation*}\label{paul}
  G_{j+1}   \left(X^{X_{\aprots}^i(\wire{w})
    \oplus X_{\bprots}^j(\wire{w})}Z^{Z_{\aprots}^j(\wire{w})\oplus
    Z_{\bprots}^j(\wire{w})}\right) 
 = \pm \left( X^{X_{\aprots}^j(\wire{w})
    \oplus X_{\bprots}^j(\wire{w})}Z^{Z_{\aprots}^j(\wire{w})\oplus
    Z_{\bprots}^j(\wire{w})}\right) G_{j+1}
\enspace.
\end{equation*}}{}
It means that up to an irrelevant phase factor, it suffices for the
owner of \wire{w}\ to apply $G_{j+1}$ without the need for neither
party to update their shared keys, i.e.,
$K_{\aprots}^{j+1}(\wire{w}):=K_{\aprots}^{j}(\wire{w})$ and
$K_{\bprots}^{j+1}(\wire{w}):=K_{\bprots}^{j}(\wire{w})$.

\fullv{
\subsubsection{\hgate,  \pgate{}, and \cnot\   on local wires.}}{}
Now, suppose  that $G_{j+1}\in \{\hgate, \pgate\}$.
Each of these one-qubit gates applied upon
wire \wire{w}\  
will be computed  by simply letting the party owning \wire{w}\
apply $G_{j+1}$. 
\fullv{Since
\begin{eqnarray*}
\hgate\, \left(  X^{X_{\aprots}^j(\wire{w})\oplus X_{\bprots}^j(\wire{w})}
          Z^{Z_{\aprots}^j(\wire{w})\oplus
          Z_{\bprots}^j(\wire{w})}\right)  
&=&  \left( X^{Z_{\aprots}^j(\wire{w})\oplus Z_{\bprots}^j(\wire{w})}
     Z^{X_{\aprots}^j(\wire{w})\oplus
          X_{\bprots}^j(\wire{w})}\right)\, \hgate 
          \enspace ,
          \mbox{ and }\\
 \pgate{}\, \left( X^{X_{\aprots}^j(\wire{w})\oplus X_{\bprots}^j(\wire{w})}
           Z^{Z_{\aprots}^j(\wire{w})\oplus
          Z_{\bprots}^j(\wire{w})}\right)
&=&  \left( X^{X_{\aprots}^j(\wire{w})\oplus X_{\bprots}^j(\wire{w})}
     Z^{X_{\aprots}^j(\wire{w})\oplus X_{\bprots}^j(\wire{w})\oplus
          Z_{\aprots}^j(\wire{w})\oplus Z_{\bprots}^j(\wire{w})}
     \right) \pgate \enspace ,
\end{eqnarray*}
the encryption keys are updated as follows:}{Encryption
keys are updated locally as:}
\begin{eqnarray*}
\hgate &:& 
K_{\aprots}^{j+1}
=
(X_{\aprots}^{j+1}(\wire{w}),Z_{\aprots}^{j+1}(\wire{w}))
:=
(Z_{\aprots}^i(\wire{w}),X_{\aprots}^j(\wire{w}))\enspace ,\\ 
& & 
K_{\bprots}^{j+1}=(X_{\bprots}^{j+1}(\wire{w}),Z_{\bprots}^{j+1}(\wire{w}))
:= 
(Z_{\bprots}^j(\wire{w}),X_{\bprots}^j(\wire{w}))\enspace ,\\
\pgate{} &:& K
_{\aprots}^{j+1}=(X_{\aprots}^{j+1}(\wire{w}),Z_{\aprots}^{j+1}(\wire{w}))
:= 
(X_{\aprots}^j(\wire{w}),X_{\aprots}^j(\wire{w})\oplus
   Z_{\aprots}^j(\wire{w}))\enspace ,\\ 
& & 
K_{\bprots}^{j+1}
=
(X_{\bprots}^{j+1}(\wire{w}), Z_{\bprots}^{j+1}(\wire{w}))
:=
(X_{\bprots}^j(\wire{w}),X_{\bprots}^j(\wire{w})\oplus 
   Z_{\bprots}^j(\wire{w}))\enspace .
\end{eqnarray*}
Any one-qubit gate in the Clifford group can be implemented the same
way using their own commutation relations with the Pauli operators
used for encryption.  A \cnot-gate on local wires can be evaluated in
a similar way.  That is, whenever both wires \wire{w}\ and \wire{w'}\
feeding the \cnot\ belong to the same party.  Assume that \wire{w}\ is
the control wire while \wire{w'}\ is the target and that \aprots\
holds them both\fullv{(i.e., $\wire{w},\wire{w'} \in \hreg{A}_0$)}{}.  Then,
\aprots\ simply applies \cnot{}\  on wires \wire{w}\ and
\wire{w'}. Encryption keys are updated as:
\begin{eqnarray*}
\cnot{} &:& 
K_{\aprots}^{j+1}(\wire{w})
=
(X_{\aprots}^{j+1}(\wire{w}), Z_{\aprots}^{j+1}(\wire{w}))
:=
(X_{\aprots}^j(\wire{w}),Z_{\aprots}^j(\wire{w})\oplus 
   Z_{\aprots}^j(\wire{w'}))\enspace ,\\ 
& & K_{\aprots}^{j+1}(\wire{w'})
=
(X_{\aprots}^{j+1}(\wire{w'}), Z_{\aprots}^{j+1}(\wire{w'}))
:=
(X_{\aprots}^j(\wire{w'})\oplus X_{\aprots}^{j}(\wire{w}),
   Z_{\aprots}^j(\wire{w'}))\ ,\\ 
& & K_{\bprots}^{j+1}(\wire{w})
:= 
K_{\bprots}^{j}(\wire{w}) \mbox{ and } K_{\bprots}^{j+1}(\wire{w'})
:= 
K_{\bprots}^{j}(\wire{w'})\enspace .
\end{eqnarray*}
When \bprots\ holds both wires, the procedure is simply performed with
the roles of \aprots\ and \bprots\ reversed.

\subsubsection{Nonlocal \cnot{}.}\label{sect:cnot}
We now look at the case where $G_{j+1}=\cnot$ 
upon wires $\wire{w}$ and $\wire{w'}$,  one of which is owned by \aprots\
while the other is owned by \bprots.
In this case, interaction is unavoidable
for the evaluation of the gate.  Let us assume w.l.g that \aprots\ holds
the control wire \wire{w}\ while \bprots\ holds the target wire \wire{w'}
(i.e., $\wire{w}\in \hreg{A}_0$ and $\wire{w'}\in \hreg{B}_0$).
We start from a construction introduced in \cite{GC99a} in the context
of fault tolerant quantum computation.

\begin{wrapfigure}{r}{0.4\textwidth}
\vspace{-20pt}
\centering
{
        \mbox{\scriptsize \Qcircuit @C=1em @R=.7em {
      \lstick{\wire{w}} & \multigate{1}{Bell} & \ustick{a_x} \cw & \cw             & \control\cw & & &\\
      \twowire{\ket{\Psi_{0,0}}}  & \ghost{Bell}        & \ustick{a_z} \cw & \control\cw     & \cwx  \\
                            & \ctrl{1}            & \qw            & \gate{Z}\cwx    & \gate{X}\cwx & \gate{Z}     & \qw \\
      \twowire{\ket{\Psi_{0,0}}}  & \targ               & \qw            & \gate{X}        & \gate{X}\cwx & \gate{Z}\cwx & \qw \\
                            & \multigate{1}{Bell} & \ustick{b_x} \cw & \control\cw\cwx &              & \cwx \\
      \lstick{\wire{w'}} & \ghost{Bell}        & \ustick{b_z} \cw & \cw             & \cw          & \control\cwx\cw                  
        \gategroup{1}{4}{6}{6}{2em}{--}
        }}
      }    
        \small{\caption{Evaluation of \cnot.}  
	\label{fig:postteleportation}}
\vspace{-10pt}
\end{wrapfigure}
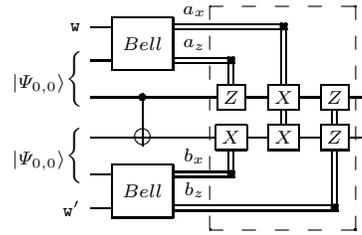

The idea behind the sub-protocol is depicted in
Fig.~\ref{fig:postteleportation}. The effect of the Bell measurement
is to {\em teleport} the input state of wires $\wire{w}$ and
$\wire{w'}$ {\em through} the \cnot-gate\cite{GC99a}.  The input to
the \cnot\ appearing in the circuit of Fig.~\ref{fig:postteleportation} is
independent of both input wires \wire{w}\ and \wire{w'} (they are just
two half \textsc{epr}-pairs).

The sub-protocol for the evaluation of \cnot\ simply consists in
executing the circuit of Fig.~\ref{fig:postteleportation} without the
decryption part (i.e., the part inside the dotted rectangle).  The
state $\ket{\xi} := (\Idt_{\reg{A}} \tensor CNOT \tensor
\Idt_{\reg{B}}) \ket{\Psi_{0,0}} \ket{\Psi_{0,0}}$ can be prepared by
one party.  We let the holder of the {\em control wire} (i.e., \aprots\
in Fig.~\ref{fig:postteleportation}) prepare $\ket{\xi}$ before
sending its two rightmost registers to the other party.  The
decryption in the dotted-rectangle is used to update the encryption
keys according to the measurement outcomes $(a_x,a_z,b_x,b_z)$:
\begin{eqnarray*}
\cnot{} &:&  
K_{\aprots}^{j+1}(\wire{w})
:=
(X_{\aprots}^j(\wire{w})\oplus a_x,Z_{\aprots}^j(\wire{w})\oplus
a_z)\enspace ,\\
& & 
K_{\bprots}^{j+1}(\wire{w}) 
:=
(X_{\bprots}^i(\wire{w}),Z_{\bprots}^j(\wire{w})\oplus b_z)\enspace ,\\
& & K_{\aprots}^{j+1}(\wire{w'}) 
:= 
(X_{\aprots}^j(\wire{w'})\oplus a_x, Z_{\aprots}^j(\wire{w'}))\enspace ,\\ 
& & 
K_{\bprots}^{j+1}(\wire{w'}) 
:= 
(X_{\bprots}^j(\wire{w'})\oplus b_x, Z_{\bprots}^j(\wire{w'})\oplus
b_z)\enspace .
\end{eqnarray*}
As for all previous gates, the key updating phase is performed
locally without the need for communication.


\subsection{Evaluation of the \rgate-Gate}\label{sect:rgate}

The only gate left in \ug\ is $G_{j+1}:=\rgate$.  We assume without
loss of generality that \aprots\ owns wire \wire{w}\ upon which
\rgate{}\ is applied (i.e., $\wire{w}\in \hreg{A}_0$).  The subprotocol
needs a call to an ideal \nl\ in order to guarantee privacy during the
key updating process. Observe first that the \rgate-gate commutes with
Pauli encryption operator $Z$.  It means that 
applying the \rgate-gate upon a state encrypted with $Z$ produces the 
correct output state still encrypted with $Z$. However, 
the equality $\rgate{}\cdot X = e^{-i\pi/4}Y\pgate{}\cdot \rgate{}$
tells us that a \pgate{}-gate should be applied for the decryption of
the output when the input has been encrypted using $X$. This breaks
the invariant that wires after each gate are all encrypted by Pauli
operators. We remove the \pgate{}-gate by converting it into a
sequence of Pauli operators.

\fullv{Suppose \aprot's wire \wire{w}\ is encrypted as usual by shared keys
$K_{\aprots}^{j}(\wire{w}) :=
(X_{\aprots}^j(\wire{w}),Z_{\aprots}^j(\wire{w}))$, and
$K_{\bprots}^{j}(\wire{w}) :=
(X_{\bprots}^j(\wire{w}),Z_{\bprots}^j(\wire{w}))$. }{}  Ignoring an irrelevant global phase, the result of
applying \rgate{}\ on wire \wire{w}\ is
\begin{align}\label{comr}
\begin{split} 
\rgate{} &
Z^{Z_{\aprots}^j(\wire{w}) \oplus Z_{\bprots}^i(\wire{w})} X^{X_{\aprots}^j(\wire{w}) \oplus
   X_{\bprots}^j(\wire{w})}
  = \\
 & Z^{Z_{\aprots}^j(\wire{w}) \oplus Z_{\bprots}^j(\wire{w})\oplus X_{\aprots}^j(\wire{w}) \oplus X_{\bprots}^j(\wire{w})}
X^{X_{\aprots}^j(\wire{w}) \oplus X_{\bprots}^j(\wire{w})}
\pgate^{X_{\aprots}^j(\wire{w}) \oplus X_{\bprots}^j(\wire{w})} \rgate{} \enspace ,
\end{split}\end{align}
\fullv{
}{}
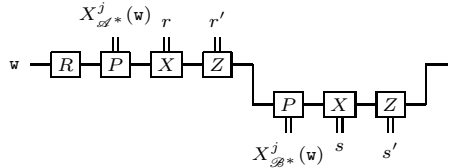
\begin{wrapfigure}{r}{0.55\textwidth} 
  \centering
  \mbox{\scriptsize \Qcircuit @C=1em @R=.7em {
                            &          & \ustick{{X_{\aprots}^j(\wire{w})}} & \ustick{r} & \ustick{r'} &      & & & & &     & \\
      \lstick{\wire{w}} & \gate{R} & \gate{P}\cwx & \gate{X}\cwx & \gate{Z}\cwx & \qw  & & & & & \qw & \\
                            &          &          &          &          & \qwx & \gate{P} & \gate{X} & \gate{Z} & \qwx \qw \\
                            &          &          &          &  &  & 
\dstick{{X_{\bprots}^j(\wire{w})}}\cwx & \dstick{s}\cwx & \dstick{s'}\cwx\\
                            & 
    } 
    }
  \small{\caption{Implementation of the \rgate{}-gate.}
  \label{fig:postrotation}}
  \vspace{-20pt}
\end{wrapfigure}
To remove the \pgate{}-gate, we let each party remove his part of
$\pgate{}^{X_{\aprots}^j(\wire{w}) \oplus X_{\bprots}^j(\wire{w})}$ in
a private interactive process.  To do this, \aprots\ picks random bits
$r$ and $r'$, and \bprots\ picks random bits $s$ and $s'$.  \aprots\
applies the operator $X^{r}Z^{r'}P^{X_{\aprots}^i(\wire{w})}$ and
sends the resulting quantum state to \bprots. \bprots\ applies the
operator $X^{s}Z^{s'}P^{ X_{\bprots}^j(\wire{w})}$ and sends the
result back to \aprots. The resulting protocol is shown in
Fig.~\ref{fig:postrotation}.  It starts with \aprots\ applying
\rgate{}\ upon the encrypted state before the one-round interactive
process described above starts.

After \aprots's application of \rgate{}, the resulting state is as
described on the right-hand side of (\ref{comr}). At the end of the
process (i.e., circuit of Fig.~\ref{fig:postrotation}), the encryption
becomes:
\begin{align}\label{lg}
\begin{split}
Z^{s'}& X^{s} \pgate{}^{X_{\bprots}^j(\wire{w})}Z^{r'} 
X^{r} \pgate{}^{X_{\aprots}^j(\wire{w})} \\
&
Z^{Z_{\aprots}^j(\wire{w}) \oplus Z_{\bprots}^j(\wire{w})\oplus X_{\aprots}^j(\wire{w}) \oplus X_{\bprots}^j(\wire{w})}
X^{X_{\aprots}^j(\wire{w}) \oplus X_{\bprots}^j(\wire{w})}
\pgate{}^{X_{\aprots}^j(\wire{w}) \oplus X_{\bprots}^j(\wire{w})}\enspace .
\end{split}
\end{align}
\fullv{Since $Z$ and \pgate{}\ commute and $\pgate\cdot X=X Z\cdot\pgate$, we 
can re-write (\ref{lg}) (i.e., up to an irrelevant phase factor)
as 
\begin{align*}
\begin{split}
Z^{s'\oplus r' \oplus r\cdot X_{\bprots}^j(\wire{w}) } & 
X^{s\oplus r}\pgate{}^{X_{\aprots}^j(\wire{w})+X_{\bprots}^j(\wire{w})} \\ 
& 
 Z^{Z_{\aprots}^j(\wire{w}) \oplus Z_{\bprots}^j(\wire{w})\oplus X_{\aprots}^j(\wire{w}) \oplus X_{\bprots}^j(\wire{w})}
X^{X_{\aprots}^j(\wire{w}) \oplus X_{\bprots}^j(\wire{w})
      X_{\aprots}^j(\wire{w}) \oplus X_{\bprots}^j(\wire{w})}\enspace .
\end{split}
\end{align*}
Using the fact that for $a,b\in\{0,1\}$, $\pgate{}^{a+b}=Z^{ab}\pgate{}^{a\oplus b}$,
the previous equation can be re-written as
\begin{equation}\label{lg3}
\begin{split} 
Z^{s'\oplus r' \oplus Z_{\aprots}^j(\wire{w})\oplus Z_{\bprots}^j(\wire{w})\oplus X_{\aprots}(\wire{w})\oplus X_{\bprots}(\wire{w})\oplus (r\oplus X_{\aprots}^j(\wire{w}))\cdot X_{\bprots}^j(\wire{w}) } & \\ 
 X^{s\oplus r} 
\pgate{}^{X_{\aprots}^j(\wire{w})\oplus X_{\bprots}^j(\wire{w})} 
X^{X_{\aprots}^j(\wire{w}) \oplus X_{\bprots}^j(\wire{w})} &
\pgate{}^{X_{\aprots}^j(\wire{w}) \oplus X_{\bprots}^j(\wire{w})}.
\end{split}
\end{equation}
Moving the leftmost \pgate{}-gate to the right results in Pauli encryption,}{Now, we use the fact that
$Z$ commutes with \pgate{}\ and $\pgate\cdot X=X Z\cdot\pgate$. In addition, since 
for $a,b\in\{0,1\}$, $\pgate{}^{a+b}=Z^{ab}\pgate{}^{a\oplus b}$ we re-write (\ref{lg}) as}
\begin{equation}\label{lg4}
\begin{split}
Z^{s'\oplus r'\oplus X_{\aprots}^j(\wire{w}) \oplus X_{\bprots}^j(\wire{w})\oplus Z_{\aprots}^j(\wire{w})\oplus Z_{\bprots}^j(\wire{w})\oplus (r\oplus X_{\aprots}^j(\wire{w}))\cdot X_{\bprots}^j(\wire{w}) } \\ X^{s\oplus r\oplus X_{\aprots}^j(\wire{w}) \oplus X_{\bprots}^j(\wire{w})} \enspace .
\end{split}
\end{equation}
Encryption (\ref{lg4}) is not a proper additive sharing since the
$Z$-operator depends on $(r\oplus X_{\aprots}^j(\wire{w}))\cdot
X_{\bprots}^j(\wire{w})$; the logical {\sc and} between a value known
only by \aprots\ (i.e., $r\oplus X_{\aprots}^j(\wire{w})$) and a value
known only by \bprots\ (i.e., $X_{\bprots}^j(\wire{w})$).
\begin{wrapfigure}{r}{0.4\textwidth}
\vspace{-15pt}\hspace{0.5in}
\centering 
  \mbox{\scriptsize \Qcircuit @C=1em @R=.7em {
      \lstick{{r\oplus X_{\aprots}^j(\wire{w})}} \ar[r] &
      \multigate{1}{\mbox{\nl}} \ar[r] & \rstick{\alpha} \\
      \lstick{X_{\bprots}^j(\wire{w})} \ar[r] & \ghost{\mbox{\nl}} \ar[r] & 
           \rstick{\beta}\\
    }}
 \small{\caption{
$\alpha\oplus\beta=(r\oplus X_{\aprots}^j(\wire{w}))\cdot X_{\bprots}^{j}(\wire{w})$ from an \nl.}
  \label{fig:resharingand}}
\vspace{-10pt}
\end{wrapfigure}
To get back to an additive sharing, \aprots\ and \bprots\ can simply
call the \nl\ once with inputs $r\oplus X_{\aprots}^j(\wire{w})$ and
$X_{\bprots}^j(\wire{w})$ respectively as depicted in
Fig.~\ref{fig:resharingand}. After this, \aprots\ and \bprots\ share a
proper encryption of the resulting state. The new encryption key for
\aprots's wire $\wire{w}$ becomes:
\begin{eqnarray*}
\rgate &:& 
K_{\aprots}^{j+1}(\wire{w})
:=
(r\oplus X_{\aprots}^j(\wire{w}),r'\oplus \alpha \oplus 
Z_{\aprots}^j(\wire{w})\oplus  X_{\aprots}^j(\wire{w}))\ ,\\ 
& & 
K_{\bprots}^{j+1}(\wire{w})
:=
(s\oplus X_{\bprots}^j(\wire{w}),s'\oplus \beta \oplus
Z_{\bprots}^j(\wire{w})\oplus X_{\bprots}^j(\wire{w}))\ .
\end{eqnarray*}

\subsection{On the Necessity of Swapping Privately}

One may ask whether relying upon \swap\ is necessary for the protocol
to be private against specious adversaries. For instance, what would happen if one party
announces the encryption keys before the other party? 
We now show that as soon as one party gets the other party's decryption key before having announced
its own, a specious adversary  can break privacy.

Consider the protocol for a quantum circuit made out of one single \cnot-gate.  
Suppose that \aprots\ holds the control wire \wire{w}\ while \bprots\
holds the target wire \wire{w'}. Suppose also the key-releasing phase first asks
 \bprots\  to announce the encryption keys $\keyb{w}{}$ before \aprots\ announces $\keya{w'}{}$. 
 Suppose \ada's input state is $\ket{0}$.
 
 The adversary \ada\ can now act as follows. \ada\ runs the protocol for \cnot\ without
 performing the Bell measurement until she receives the encryption key $b_z$ from \bprots. 
 Clearly, \ada's behavior is specious up to that point since she could re-produce the 
 honest state by just applying the Bell measurement on her input state \fullv{stored in register
 $\hreg{A}_0$}{}.  
 However,  given $b_z$ she could also
 in principle compute the \cnot\ upon any input state of her choice. This means that 
 the state she holds after $b_z$ has been announced and before applying her Bell measurement
 contains information about \bprots's input.
%
%
 On the one hand, 
 when \ada's input state is $\ket{0}$ no information whatsoever on \bprots's
 input state should be available to her (i.e., in this case \cnot\ behaves like the
 identity). On the other hand, had her input state been 
 $\ket{-}$,  information about  \bprots's state would have become available since
 the control and target wires exchange their roles when the input states are in the Hadamard basis. 
 However,
 when \ada's input state is $\ket{0}$, 
 any simulation of her view can only call the ideal functionality  with
 input state $\ket{0}$.\delete{and  no information about \bprots's state
 is available to the simulator.} It follows that no simulator can 
 reproduce   \ada's state right after the announcement
 of $b_z$.

 \fullv{}{ 
 \section{Proof of Privacy} 
 \subsubsection{Privacy of the Evaluation Phase.}
 We start by showing privacy of protocol $\pup{U}^{\qop{O}}=(\aprots,\bprots,n_U+1)$
 at all steps $1\leq i \leq n_U-1$ occurring during 
the {\em evaluation  phase}
of quantum  
circuit $C_U$ implementing $U$ with $u$ gates in \ug. 
The last step of the evaluation phase is $n_U-1$ since only
one oracle call is left to complete the execution.
This phase is the easy part of the simulation since 
all transmissions  are
independent of the joint input state $\rhoin \in \dens{\hreg{A}_0\otimes\hreg{B}_0\otimes \hreg{R}}$.
The lemma below can easily be proven and provides a 
perfect simulation of  any adversary's view generated
during the evaluation of any gate in $C_U$. No call to the ideal functionality for $U$
is required.

\begin{lemma}\label{induct} 
$\pup{U}^{\qop{O}}=(\aprots,\bprots,n_U+1)$
admits a simulator $\simulator{\ada}$ for any adversary \ada\ (not necessarily specious)
that does not
call the ideal functionality for $U\in \unit{\hreg{A}_0 \otimes \hreg{B}_0}$ such that
for any joint input state $\rhoin \in \dens{\hreg{A}_0 \otimes\hreg{B}_0\otimes \hreg{R}}$, 
every $1\leq i \leq n_U-1$:
\begin{equation*}
\Delta\left(\view{\ada}{\rhoin}{i},\trace[\hreg{B}_i]{\astatei{i}{\ada}{\rhoin}} \right) = 0 \enspace .
\end{equation*}
The same holds against any adversary \adb.
\end{lemma}
 
 \subsubsection{Privacy of the Key-Releasing Phase.}
 Before proving privacy of the key-releasing phase, we need the following lemma 
 establishing that at the end of the protocol, 
specious adversaries must leave their extra working registers (used to implement the attack) 
independent of the joint input state.
In other words, no extra information is available to the adversary at the very end of any correct  protocol. Hence, if the adversary can break the privacy of a protocol, then he must ``rush'' to do so before the last step.

\begin{lemma}[Rushing Lemma]\label{rushing}
	Let $\puq{U} = (\aprot,\bprot, n)$ be a correct protocol for the two party evaluation of $U$.  Let \ada\ be any $\varepsilon$--specious adversary in $\puq{U}$. Then, there exists an isometry $T : \tilde{\hreg{A}}_n \rightarrow \hreg{A}_n \otimes \widehat{\hreg{A}}$ and a mixed state $\tilde{\varrho} \in \dens{\widehat{\hreg{A}}}$ such that for all joint input states $\rhoin \in \dens{\hreg{A}_0\otimes\hreg{B}_0 \otimes \hreg{R}}$,
\begin{equation*}
	\Delta\left((T \otimes \Idt_{\hreg{B}_n \otimes \hreg{R}})\cdot\left([\ada\compose\bprot](\rhoin)\right),\tilde{\varrho}\otimes 
	(U \otimes \Idt_{\hreg{R}})\cdot\rhoin \right)\leq 12 \sqrt{2\varepsilon} \enspace .
\end{equation*} 
The same also applies to any $\varepsilon$--specious adversary $\adb$.
\end{lemma}

\begin{proof}
	We shall only prove the statement for an $\varepsilon$--specious \ada; the statement for an $\varepsilon$--specious \adb\ is identical. Furthermore, by convexity, it is sufficient to prove the theorem for pure $\rhoin$.

Consider any pair of pure input states $\ket{\psi_1}$ and $\ket{\psi_2}$ in $\hreg{A}_0 \otimes \hreg{B}_0 \otimes \hreg{R}$. Now, let $\hreg{R}' := \hreg{R} \otimes \hreg{R}_2$, where $\hreg{R}_2 = \Span\{\ket{1}, \ket{2}\}$ represents a single qubit, and define the state $\ket{\psi} := \frac{1}{\sqrt{2}}(\ket{\psi_1}\ket{1} + \ket{\psi_2}\ket{2}) \in \hreg{A}_0 \otimes \hreg{B}_0 \otimes \hreg{R}'$. Note that $\trace[\hreg{R}_2]{\proj{\psi}} = \frac{1}{2} \ket{\psi_1}\bra{\psi_1} + \frac{1}{2} \ket{\psi_2}\bra{\psi_2}$. Due to the correctness of the protocol and to the speciousness of $\ada$, there exists a quantum operation $\trans{n} : \linop{\tilde{\hreg{A}_n}} \rightarrow \linop{ \hreg{A}_n}$ such that
\begin{equation*}
	\Delta\left( (\trans{n} \otimes \Idt_{\linop{\hreg{B}_n \otimes \hreg{R}'}})([\ada \compose \bprot](\proj{\psi})) , (U \otimes \Idt_{\hreg{R}'})\cdot \proj{\psi}  \right) \leq 2\varepsilon\enspace .
\end{equation*}
Now, consider any isometry $T : \tilde{\hreg{A}}_n \rightarrow \hreg{A}_n \otimes \widehat{\hreg{A}}$ such that $\trans{n}(\sigma) = \trace[\widehat{\hreg{A}}]{T \sigma T^\dagger}$ for every $\sigma \in \linop{\tilde{\hreg{A}}_n}$ --- in other words, any operation that implements $\trans{n}$ while keeping any information that would otherwise be destroyed in $\widehat{\hreg{A}}$. By Uhlmann's theorem, there must exist a state $\tilde{\varrho} \in \dens{\widehat{\hreg{A}}}$ such that
\begin{equation*}
	\Delta\left((T \otimes \Idt_{\hreg{B}_n \otimes \hreg{R}'})\cdot \left([\ada\compose\bprot](\ket{\psi}\bra{\psi})\right),
	\tilde{\varrho}\otimes \left((U \otimes \Idt_{\hreg{R}'})\cdot\proj{\psi}\right) \right)\leq 2 \sqrt{2 \varepsilon}\enspace .
\end{equation*}
Now, the trace distance is monotonous under completely positive, trace non-increasing maps. In particular, we can apply the projector $P_1 = \Idt_{\linop{\hreg{A}_n \otimes \hreg{B}_n \otimes \hreg{R}}} \otimes \ket{1}\bra{1}$ to both states in the above trace distance and the inequality will still hold. In other words, we project both states onto $\ket{1}$ on $\hreg{R}_2$, thereby turning $\ket{\psi}\bra{\psi}$ into $\frac{1}{2} \ket{\psi_1}\bra{\psi_1}$. Factoring out the $\frac{1}{2}$, we get that
\begin{equation*}
	\Delta\left((T \otimes \Idt_{\hreg{B}_n \otimes \hreg{R}})\cdot\left([\ada\compose\bprot](\proj{\psi_1})\right),
	\tilde{\varrho}\otimes \left( (U \otimes \Idt_{\hreg{R}})\cdot \proj{\psi_1} \right) \right)\leq 4 \sqrt{2 \varepsilon}\enspace .
\end{equation*} 
Likewise, projecting onto $\ket{2}$ yields
\begin{equation*}
	\Delta\left((T \otimes \Idt_{\hreg{B}_n \otimes \hreg{R}})\cdot\left([\ada\compose\bprot](\proj{\psi_2})\right),
	\tilde{\varrho}\otimes \left( (U \otimes \Idt_{\hreg{R}})\cdot \proj{\psi_2} \right) \right)\leq 4 \sqrt{2 \varepsilon}\enspace .
\end{equation*} 
Our only problem at this point is that $\tilde{\varrho}$ in principle depends on $\ket{\psi_1}$ and $\ket{\psi_2}$. However, repeating the above argument with $\ket{\psi_1}$ and $\ket{\psi_3}$ for any $\ket{\psi_3}$ will yield a $\tilde{\varrho}'$ with 
\begin{equation*}
	\Delta\left((T \otimes \Idt_{\hreg{B}_n \otimes \hreg{R}})\cdot \left([\ada\compose\bprot](\ket{\psi_1}\bra{\psi_1})\right),
	\tilde{\varrho}'\otimes \left( (U \otimes \Idt_{\hreg{R}})\cdot \proj{\psi_1} \right) \right)\leq 4 \sqrt{2 \varepsilon}
\end{equation*} 
and hence, by the triangle inequality, $\Delta(\tilde{\varrho}, \tilde{\varrho}') \leq 8 \sqrt{2\varepsilon}$. Therefore, for any state $\ket{\varphi} \in \hreg{A}_0 \otimes \hreg{B}_0 \otimes \hreg{R}$, there exists a state $\tilde{\rho} \in \widehat{\hreg{A}}$ with $\Delta(\tilde{\rho}, \tilde{\varrho}) \leq 8 \sqrt{2\varepsilon}$ such that
\begin{equation*}
	\Delta\left((T \otimes \Idt_{\hreg{B}_n \otimes \hreg{R}})\cdot \left([\ada\compose\bprot](\ket{\varphi}\bra{\varphi})\right),
	\tilde{\rho}\otimes \left( (U \otimes \Idt_{\hreg{R}})\cdot \proj{\varphi} \right) \right)\leq 4 \sqrt{2 \varepsilon} \enspace .
\end{equation*} 
The lemma then follows by the triangle inequality:
\begin{multline*}
	\Delta\left((T \otimes \Idt_{\hreg{B}_n \otimes \hreg{R}})\cdot \left([\ada\compose\bprot](\proj{\varphi}\right),
	\tilde{\varrho}\otimes \left( (U \otimes \Idt_{\hreg{R}})\cdot \proj{\varphi} \right) \right)\\
	\begin{split}
		&\leq \Delta\left((T \otimes \Idt_{\hreg{B}_n \otimes \hreg{R}})\cdot \left([\ada\compose\bprot](\proj{\varphi})\right),
		\tilde{\rho}\otimes \left( (U \otimes \Idt_{\hreg{R}})\cdot \proj{\varphi} \right) \right) + \Delta(\tilde{\rho}, \tilde{\varrho})\\
		&\leq 4 \sqrt{2\varepsilon} + 8 \sqrt{2\varepsilon} = 12 \sqrt{2\varepsilon} \enspace .
	\end{split}
\end{multline*}
\qed
\end{proof}

In order to conclude the privacy of $\pup{U}^{\qop{O}}$, families $\simulator{\ada}$
and $\simulator{\adb}$  
need one more simulator each: $\simul{n_U}\in  \simulator{\ada}$
and $\simul{n_U}' \in  \simulator{\adb}$  corresponding
to the simulation of the key-releasing phase. 
This time, these simulators 
need to query the ideal functionality for $U$ and also
need the adversary to be specious.  
We show that privacy of the key-releasing phase 
follows from the ``Rushing Lemma'' (Lemma~\ref{rushing}).
This is the role of the ideal \swap\ to make sure that 
before the adversary  gets the output of the computation, the information
needed by the honest player to recover its own output has been given away
by the adversary. 

It should be mentioned that we're not explicitly simulating the 
final state of the adversary since simulating the
\swap\ allows also to get \ada's final state  
by simply adding \ada's last quantum operation to the simulated
view. We therefore set step $n_U$ in $\pup{U}^{\qop{O}}$
to be the step reached after the call to \swap. This abuses the notation
a  bit since after \swap, \ada\ and \bprots\ must each apply 
a final quantum operation with no more oracle call. We'll denote
by $\ada_{n_U+1}$ and $\bprots_{n_U+1}$ these last operations
allowing to reconstruct the output of the computation (no comunication).

\begin{lemma}\label{release}
For any $\varepsilon$-specious quantum adversary \ada\ 
against $\pup{U}^{\qop{O}}=(\aprots,\bprots,n_U+1)$, there 
exist  simulators $\simul{n_U}\in \simulator{\ada}$ 
such that for all 
$\rhoin \in \dens{\hreg{A}_0\otimes \hreg{B}_0\otimes\hreg{R}}$,
\begin{equation*}
\Delta\left(\view{\ada}{\rhoin}{n_U},\trace[\hreg{B}_{n_U}]{\astatei{n_U}{\ada}{\rhoin}} \right) \leq  
24\sqrt{2\varepsilon} \enspace .
\end{equation*}
Simulator $\simul{n_U}$ 
calls the ideal functionality for $U$
and can be used directly to simulate  step $n_U+1$ as well. The same holds
for adversary \adb.
\end{lemma}

\begin{proof}[sketch]
We only prove privacy against adversary \ada, 
privacy against \adb\ follows directly since the key-releasing phase
is symmetric.  The idea behind the proof is to run 
\ada\ and \bprots\ upon a dummy joint input state until the end
of the protocol. Since the adversary is specious, it can re-produce
the honest state at the end. The Rushing Lemma tells us that
at this point, the output of the computation is essentially in tensor
product with all the other registers. Moreover, the state of all other
registers is independent of the input state upon which the protocol
is executed. The {\em dummy output} can then be replaced
by the output of the ideal functionality for $U$ before \ada\
goes back to the stage reached just after \swap. 


More formally,
we define a simulator $\simul{n_U}\in \simulator{\ada}$ producing 
\ada's view just after the call to \swap. 
Let $\ada_{\swap}\in \lin{\hreg{A}_0}{\tilde{\hreg{A}}_{n_U}}$ and 
$\bprots_{\swap}\in \lin{\hreg{B}_0}{\tilde{\hreg{B}}_{n_U}}$ be the 
quantum operations run by \ada\ and \bprots\  respectively until after \swap\ is executed.
Notice that at this point, \ada's and \bprots's registers
do not  have any further oracle registers  
since no more communication or oracle call will take place.
Let $\tilde{A}_{n_U}\in\lin{\tilde{\hreg{A}}_{n_U}}{\tilde{\hreg{A}}_{n_U+1}\otimes \hreg{Z}}$ be the isometry 
implementing \ada's  last quantum operation  taking
place after the call to \swap\ (and producing her final state) and
let   ${B}_{n_U}\in\lin{{\hreg{B}}_{n_U}}{{\hreg{B}}_{n_U+1}\otimes \hreg{W}}$ be the isometry 
implementing \bprots's  last quantum operation.  Finally,
let $T\in \lin{\tilde{\hreg{A}}_{n_U+1}}{\hreg{A}_{n_U+1}\otimes\hat{\hreg{A}}}$ be the isometry
implementing $\trans{n_U+1}$ as defined in Lemma~\ref{rushing}
(i.e., the transcript produced at the very end of the protocol). 
As usual, let $\rhoin \in \dens{\hreg{A}_0\otimes\hreg{B}_0\otimes\hreg{R}}$ 
be the joint input state. 
The simulator $\simul{n_U}$  performs the following operations:
\begin{enumerate}
\item $\simul{n_U}$ generates 
the quantum state $\sigma(\phi^*)=[\ada_{\swap}\compose\bprots_{\swap}](\proj{\phi^*})
\in \dens{\tilde{\hreg{A}}_{n_U}\otimes \hreg{B}_{n_U}}$
 implementing 
\ada\ interacting with \bprots\ until \swap\  is applied. The execution
is performed  upon a predetermined (dummy) 
arbitrary input state
$\ket{\phi^*} \in \hreg{A}_0\otimes\hreg{B}_0$. \label{s1}

\item $\simul{n_U}$ sets $\sigma'(\phi^*) = (T \tilde{A}_{n_U}\otimes {B}_{n_U})\cdot
\sigma(\phi^*) \in \dens{\hreg{A}_{n_U+1}\otimes \hreg{B}_{n_U+1}
\otimes \hreg{Z}\otimes\hat{\hreg{A}}\otimes\hreg{W}}$.\label{s2}

\item \simul{n_U}\ replaces register $\hreg{A}_{n_U+1}\approx \hreg{A}_0$ 
by \aprots's output
of the ideal functionality for $U$ evaluated upon \rhoin.
That is, \simul{n_U}\ generates the state
$\sigma'(\rhoin)= (U\otimes \Idt_{\hreg{R}})\cdot\rhoin
\otimes \trace[\hreg{A}_{n_U+1}\hreg{B}_{n_U+1}]{\sigma'(\phi^*)}
\in \dens{\hreg{A}_{n_U+1}\otimes \hreg{B}_{n_U+1}\otimes\hreg{R}
\otimes \hreg{Z}\otimes\hat{\hreg{A}}\otimes\hreg{W}}$.
\label{s3}

\item \simul{n_U}\ finally sets $\view{\ada}{\rhoin}{n_U}=\trace[\hreg{B}_{n_U+1}\hreg{W}]
{(T \tilde{A}_{n_U}\otimes \Idt_{\hreg{B}_{n_U+1}\hreg{R}})^{\dagger}\cdot \sigma'(\rhoin)}
\in \dens{\tilde{\hreg{A}}_{n_U}\otimes\hreg{R}}$.
 \label{s4}
\end{enumerate}
Notice that execution of  the ideal \swap\ ensures that 
the keys swapped are independent of each other
and of the joint input state $\rhoin$. This is because
for any input state, all these keys are uniformly distributed
bits if they are outcomes of Bell measurements and otherwise
are set to $0$.  By Lemma~\ref{rushing} and the fact
that \ada\ is $\varepsilon$--specious,
we have:
\[ 
\begin{split}
 \Delta\left( \trace[\hreg{Z}\hat{\hreg{A}}\hreg{W}]{\sigma'(\phi^*)}, \tilde{\varrho}\otimes U\cdot \proj{\phi^*}\right) \leq 
12\sqrt{2\varepsilon} \mbox{ and \hspace{0in}}  &    \\ 
 \Delta\left( (\trans{n_U+1}\otimes\Idt_{\linop{\hreg{B}_{n_U+1}}}) \left([\ada\compose \bprots](\rhoin)\right), \tilde{\varrho}\otimes U\cdot\rhoin\right)  \leq 
12\sqrt{2\varepsilon}\enspace .
\end{split}
\]
It follows using the triangle inequality  that,
\begin{equation}\label{wow}
 \Delta\left( (\trans{n_U+1}\otimes\Idt_{\linop{\hreg{B}_{n_U+1}}}) \left([\ada\compose \bprots](\rhoin)\right),   
\trace[\hreg{Z}\hat{\hreg{A}}\hreg{W}]{\sigma'(\rhoin)}
\right) \leq 24 \sqrt{2\varepsilon}\enspace .
\end{equation}
Using the fact that isometries cannot increase the trace-norm distance and
that $(T\tilde{A}_{n_U})^{\dagger}$ allows \ada\  to go back from the end of the protocol
to the step reached after \swap, 
we get from (\ref{wow}) that
\[
\begin{split}
&\Delta\left(\view{\ada}{\rhoin}{n_U},\trace[\hreg{B}_{n_U}]{\astatei{n_U}{\ada}{\rhoin}} \right)
 = \\ & \Delta\left(
(\trans{n_U+1}\otimes\Idt_{\linop{\hreg{B}_{n_U+1}}}) \left([\ada\compose \bprots](\rhoin)\right),   
\trace[\hreg{Z}\hat{\hreg{A}}\hreg{W}]{\sigma'(\rhoin)}
\right) 
\leq 24\sqrt{2\varepsilon}   \enspace .
\end{split}\]
The proof of the statement follows.
\qed
\end{proof}}

\section{Main Result and Open Questions}

Putting Lemma~\ref{induct} and Lemma~\ref{release} together gives
the desired result:

\begin{theorem}[Main Result]\label{final}
  Protocol $\pup{U}^{\qop{O}}$ is
  statistically private against any statistically specious quantum
  adversary and for any $U\in \unit{\hreg{A}_0\otimes
  \hreg{B}_0}$. If $U$ is in the Clifford group then the only
  non-trivial oracle call in \qop{O}\ is one call to an ideal \swap.
  If $U$ is not in the Clifford group then \qop{O}\ contains an
  additional oracle call to \nl\ for each \rgate-gate in the circuit
  for $U$.
\end{theorem}

It should be mentioned that it is not too difficult to modify our
protocol in order to privately evaluate quantum operations rather than
only unitary transforms. Classical two party computation together with
the fact that quantum operations can be viewed as unitaries acting in
larger spaces can be used to achieve this extra functionality. Privacy
can be preserved by keeping these extra registers encrypted after the
execution of the protocol. We leave this discussion to the full version of the paper.

A few interesting questions remain open:
\begin{itemize}


\item 

  It would be interesting to know whether there exists a unitary
  transform that can act as a universal primitive for private
  two-party evaluation of unitaries.  This would allow to determine
  whether classical cryptographic assumptions are required for this
  task.

\item 

  Finally, is there a way to compile quantum protocols secure against
  specious adversaries into protocols secure against arbitrary
  quantum adversaries? An affirmative answer would allow to simplify
  greatly the design of quantum protocols. Are extra assumptions needed
  to preserve privacy against any adversary? 

\end{itemize}

\section{Acknowledgements}

The authors would like to thank the referees for their comments
and suggestions. 
We would also like to thank Thomas Pedersen for numerous 
helpful discussions in the early stage of this work.


\addcontentsline{toc}{chapter}{Bibliography} 
\bibliographystyle{plain}
\bibliography{crypto,qip,procs}


\begin{appendix}

\fullv{\section{Commutations Rules}\label{commut}

\begin{equation*}
    X =\left[ \begin{array}{cc} 0 & 1 \\ 1 & 0 \end{array} 
      \right]\ ,\ \ 
    Y = \left[ \begin{array}{cc} 0 & -1 \\ 1 & 0 \end{array} 
      \right]\ ,\ \ 
    Z = \left[ \begin{array}{cc} 1 & 0 \\ 0 & -1 \end{array} 
      \right]\ ,
\end{equation*}
\begin{equation*}
    P =\left[ \begin{array}{cc} 1 & 0 \\ 0 & i \end{array} 
      \right]\ ,\ \ 
    H = \frac1{\sqrt2} \left[ \begin{array}{cc} 1 & 1 \\ 1 & -1 \end{array} 
      \right]\ ,\ \ 
    R = \left[ \begin{array}{cc} 1 & 0 \\ 0 & e^{i \pi/4} \end{array} 
      \right]\ ,
\end{equation*}
\begin{equation*}
  CNOT = 
  \left[ 
    \begin{array}{cccc} 
      1 & 0 & 0 & 0 \\ 
      0 & 1 & 0 & 0 \\
      0 & 0 & 0 & 1 \\
      0 & 0 & 1 & 0 
    \end{array} 
    \right]\ .
\end{equation*}

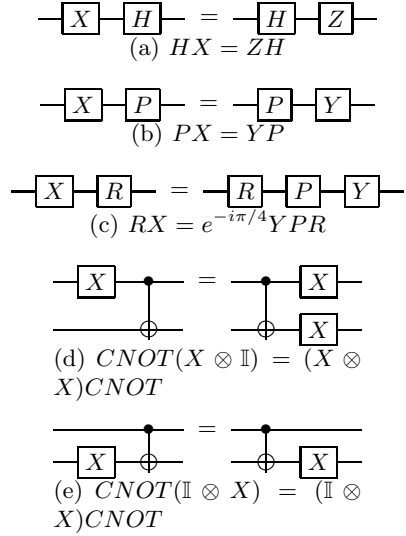
\begin{figure}[h]
  \centering
  \subfigure[$HX=ZH$]{\Qcircuit @C=1em @R=.7em {
      & \gate{X} & \gate{H} & \qw & = & & \gate{H} & \gate{Z} & \qw
    }} \\
  \subfigure[$PX=YP$]{\Qcircuit @C=1em @R=.7em {
      & \gate{X} & \gate{P} & \qw & = & & \gate{P} & \gate{Y} & \qw
    }} \\
  \subfigure[$RX=e^{-i\pi/4}YPR$]{\Qcircuit @C=1em @R=.7em {
      & \gate{X} & \gate{R} & \qw & = & & \gate{R} & \gate{P} & \gate{Y} & \qw
    }} \\
  \subfigure[$CNOT(X\tensor\id) = (X\tensor X)CNOT$]{\Qcircuit @C=1em @R=.7em {
      & \gate{X} & \ctrl{1} & \qw & = & & \ctrl{1} & \gate{X} & \qw \\
      & \qw      & \targ    & \qw &   & & \targ    & \gate{X} & \qw
    }} \\
  \subfigure[$CNOT(\id\tensor X) = (\id\tensor X)CNOT$]{\Qcircuit @C=1em @R=.7em {
      & \qw      & \ctrl{1} & \qw & = & & \ctrl{1} & \qw      & \qw \\
      & \gate{X} & \targ    & \qw &   & & \targ    & \gate{X} & \qw
    }} 
  \caption{Commutation relations for $X$.}
  \label{fig:xcom}
\end{figure}

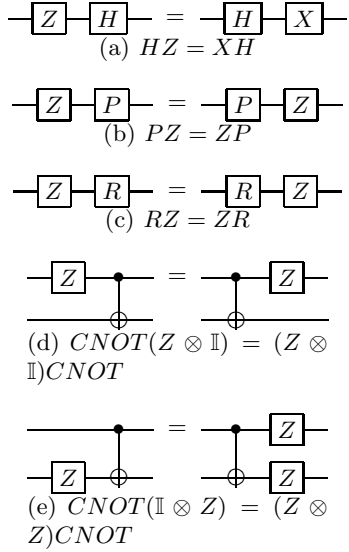
\begin{figure}[h]
  \centering
  \subfigure[$HZ=XH$]{\Qcircuit @C=1em @R=.7em {
      & \gate{Z} & \gate{H} & \qw & = & & \gate{H} & \gate{X} & \qw
    }} \\
  \subfigure[$PZ=ZP$]{\Qcircuit @C=1em @R=.7em {
      & \gate{Z} & \gate{P} & \qw & = & & \gate{P} & \gate{Z} & \qw
    }} \\
  \subfigure[$RZ=ZR$]{\Qcircuit @C=1em @R=.7em {
      & \gate{Z} & \gate{R} & \qw & = & & \gate{R} & \gate{Z} & \qw
    }} \\
  \subfigure[$CNOT(Z\tensor\id) = (Z\tensor\id)CNOT$]{\Qcircuit @C=1em @R=.7em {
      & \gate{Z} & \ctrl{1} & \qw & = & & \ctrl{1} & \gate{Z} & \qw \\
      & \qw      & \targ    & \qw &   & & \targ    & \qw      & \qw
    }} \\
  \subfigure[$CNOT(\id\tensor Z) = (Z \tensor Z)CNOT$]{\Qcircuit @C=1em @R=.7em {
      & \qw      & \ctrl{1} & \qw & = & & \ctrl{1} & \gate{Z} & \qw \\
      & \gate{Z} & \targ    & \qw &   & & \targ    & \gate{Z} & \qw
    }} 
  \caption{Commutation relations for $Z$.}
  \label{fig:zcom}
\end{figure}
}{}

\fullv{
\section{Classical Definition of a Specious Adversary} \label{app:spec}

In this section we briefly discuss the definition of an specious
adversary and the definition of security against such an adversary,
and we compare it to the notion of a semi-honest classical adversary
to illustrate the difference. 

\subsection{Specious Adversary}

As usual we let an $n$-party function $(y_1,\ldots,y_n) =
f(x_1,\ldots,x_n)$ define $n$ functions $y_i = f_i(x_1,\ldots,x_n)$.

For our purpose, an $n$-party protocol $\pi = (\pi_1,\ldots,\pi_n)$
consists of $n$ parties $\pi_i$ connected by secure channels.  If the
protocol is for the $h$-hybrid model, for an $n$-party function $h$,
there are additionally some designated rounds where each $\pi_i$ must
specify an input $a_i$ to $h$. Then $(b_1,\ldots,b_n) =
h(a_1,\ldots,a_n)$ is computed and each $\pi_i$ is given back $b_i$. A
receiving point in a protocol is a point where the parties just
exchanged messages or just received outputs $b_i$ from $h$.

For an $n$-party protocol $\pi$ and for $H \subset \{ 1, \ldots, n\}$
we denote by $\pi_H$ the set $\{ \pi_i \}_{i \in H}$ of parties
indexed by $i \in H$.

For an $n$-party protocol $\pi$ and for $C \subset \{ 1, \ldots, n\}$
we denote by $\tilde{\pi}_C$ an adversary for $\pi$ acting on behalf
of parties indexed by $i \in C$. It receives the inputs, randomness
and messages of all parties indexed by $i \in C$ and decides what
messages they should send. By$ (\pi_{\bar{C}},\tilde{\pi}_C)$ we mean
the protocol consisting of the parties $\pi_i$, $i \not\in C$, running
with the adversary $\tilde{\pi}_C$.

We use the following notation for vectors. We sometimes identify a
vector $v = (v_1,\ldots,v_n)$ with the set $\{ (i,v_i) \}_{i \in
\{1,\ldots,n\}}$. For $S \subset \{ 1, \ldots, n \}$ we let $v_S$ be
the vector $v$ restricted to indices in $S$, formally $v_S = \{
(i,v_i) \}_{i \in S}$.  For $S_1,S_2 \subset \{ 1, \ldots, n \}$ with
$S_1 \cap S_2 = \emptyset$ we let $(v_{S_1}, v_{S_2}) = v_{S_1} \cup
v_{S_2}$.

\begin{definition}[execution of (corrupted) protocol]
  For an $n$-party protocol $\pi$ and input $x = (x_1,\ldots,x_n)$,
  the distribution $\pi(x)$ is defined as follows: sample $r =
  (r_1,\ldots,r_n)$ uniformly at random. Run $\pi$ on input $x$ and
  randomness $r$. Let $y = (y_1,\ldots,y_n)$, where $y_i$ is the
  output of party $\pi_i$, and let $\pi(x) = (x,y)$.  For an $n$-party
  protocol $\pi$, input $(x_1,\ldots,x_n)$, subset $C \subset \{ 1,
  \ldots, n \}$, adversary $\tilde{\pi}_C$ and $\tilde{\pi} =
  (\pi_{\bar{C}},\tilde{\pi}_C)$, the distribution $\tilde{\pi}(x)$ is
  defined as follows: Sample $\tilde{r}_i$, $i \in C$, uniformly at
  random.  Sample $r_i$, $i \not\in C$, uniformly at random.  Let
  $\tilde{\pi} = (\pi_{\bar{C}},\tilde{\pi}_C)$. Run $\tilde{\pi}$ on
  input $x$ and randomness $(r_{\bar{C}},\tilde{r}_C)$. Let
  $\tilde{y}_C$ be the output of the adversary, let $y_{\bar{C}}$ be
  the outputs of the parties $\pi_{\bar{C}}$, and let $\tilde{\pi}(x)
  = (x,(y_{\bar{C}},\tilde{y}_C))$.
\end{definition}

\begin{definition}[specious adversary]
  Let $\pi$ be an $n$-party protocol, let $C \subset \{1,\ldots,n\}$,
  let $\tilde{\pi}_C$ be an adversary, let $\tilde{\pi} =
  (\pi_{\bar{C}},\tilde{\pi}_C)$. We say that $\tilde{\pi}_C$ is
  specious in $\pi$ if there exists a poly-time view simulator $V$
  such that for all inputs $x = (x_1,\ldots,x_n)$  and for all
  receiving points $p$ in $\tilde{\pi}$ it holds that $D^{(p)}$ and
  $\tilde{D}^{(p)}$ have the same distribution, where the distribution
  $D^{(p)}$ is defined as follows: sample $r = (r_1,\ldots,r_n)$
  uniformly at random. Run $\pi$ on input $x$ and randomness $r$ until
  receiving point $p$. Let $M = (M_1,\ldots,M_n)$, where $M_i$ is the
  messages sent and received by party $\pi_i$, and let $D^{(p)} =
  (x,r,M)$.  The distribution $\tilde{D}^{(p)}$ is defined as follows:
  Sample $r_i$, $i \not\in C$, uniformly at random. Sample
  $\tilde{r}_i$, $i \in C$, uniformly at random.  Run $\tilde{\pi}$ on
  input $x$ and randomness $(r_{\bar{C}},\tilde{r}_C)$ until receiving
  point $p$. Let $\tilde{M}_C$ be the messages sent and received by
  the adversary $\tilde{\pi}_C$, let $M_{\bar{C}}$ be the messages
  sent and received by parties $\pi_{\bar{C}}$, let $(r_C,M_C) =
  V(p,x_C,\tilde{r}_C,\tilde{M}_C)$, and let $\tilde{D}^{(p)} =
  (x,(r_C,r_{\bar{C}}),(M_C,M_{\bar{C}}))$.
\end{definition}

\begin{definition}[specious security]
  Let $\pi$ be an $n$-party protocol and let $f$ be an $n$-party
  function. By $\delta^f$ we denote the dummy protocol for $f$: it
  runs in the $f$-hybrid model and party $\delta^f_i$ on input $x_i$
  sends $x_i$ to $f$, waits for the output $y_i$ from $f$, outputs
  $y_i$ and terminates. We say that $\pi$ is a specious implementation
  of $f$ against corruptions from adversary structure $\mathcal{C}$ if
  for all $C \in \mathcal{C}$ and all adversaries $\tilde{\pi}_C$
  which are specious in $\pi$ there exists an adversary
  $\tilde{\delta}^f_C$ which is specious in $\delta^f$ such that
  $(\pi_{\bar{C}},\tilde{\pi}_C)(x) =
  (\delta^f_{\bar{C}},\tilde{\delta}^f_C)(x)$.
\end{definition}

The adversary $\tilde{\delta}^f_C$ is also called the simulator. It
gets the input $x_C$ and can then choose alternative inputs
$x_C'$. Then it receives $y'_C$, where $y' = f(x_{\bar{C}},x'_C)$, and
outputs some $\tilde{y}_C$. In the dummy protocol, there is only one
receiving point, namely after the ideal evaluation of $f$. So, for
$\tilde{\delta}^f_C$ to be specious in $\delta^f$ it needs only be
able to compute the correct view at this point. The correct view is
$y_C$ for $y = f(x)$, so a specious $\tilde{\delta}^f_C$ (in
$\delta^f$) can by definition compute $y_C$ from $x_C$, $x'_C$ and
$y_C'$ (and its own randomness if it is randomized). In words, being
specious in the ideal process means that for all inputs $x$ you give
an alternative input to $f$ which allows to reconstruct the right
output.
 
Note that if we consider an $n$-party function $f$ where all parties
receive the same output, $f_i = f_j$, then it is clear that for
$\tilde{\delta}^f_C$ to be specious it should hold that
$f(x_{\bar{C}},\tilde{x}_C) = f(x)$ for all inputs $x$, as
$f(x_{\bar{C}},\tilde{x}_C)$ is included in the messages received by
$\delta^f_{\bar{C}}$. In words, for a function $f$ with common output,
being specious in the ideal process means that for all inputs $x$ you
give an alternative input to $f$ which makes $f$ give the right
output; You can therefore only make insignificant changes to your true
input.

\subsection{Specious Adversaries can be Stronger than Semi-Honest Adversaries}

In some settings a specious adversary is strictly stronger than a
semi-honest adversary. We demonstrate this by first giving a protocol
for one-out-of-two oblivious transfer (OT) which is secure against a
poly-time semi-honest adversary, but insecure against a poly-time
specious adversary. We then show that there exists a function $f$ and
a protocol $\pi$ which is a perfectly secure implementation of $f$
against an unbounded semi-honest adversary in the OT-hybrid model, but
insecure against even a poly-time specious adversary. The first
example exploits that a specious adversary can prepare its randomness
in any way it wants. The second example exploits that a specious
adversary can provide any input it wants to ideal functionalities (in
our case the OT's of the OT-hybrid model) as long as it can later make
it look as if it gave the right input.

\begin{theorem}
  Under the computational assumption given below, there exists a
  protocol which is a secure implementation of oblivious transfer
  against a static, poly-time semi-honest adversary but which is insecure
  against a static, poly-time specious adversary.
\end{theorem}

Assume that we have a family of trapdoor permutations, where the
description of a random permutation is a random string. More formally:
\begin{itemize}
\item

  on input $n \in \mathbb{N}$ the generator $G$ outputs $(i,t)$, where
  $i$ is uniformly random in some $\zo^\ell$, and $G$ runs in
  poly-time in $n$.

\item

  Each index $i \in \zo^\ell$ defines a permutation $p_i: \zo^n
  \rightarrow \zo^n$. Given $i \in \zo^\ell$ and $x \in \zo^n$ one can
  compute $y = p_i(x)$ in poly-time in $n$. 

\item

  Given $t$, where $(i,t) \leftarrow G(n)$ and $y \in \zo^n$ one can
  compute $x = p_i^{-1}(y)$ in poly-time in $n$.

\item

  It holds for all poly-time algorithms $A$ that the probability that
  it outputs $p_i^{-1}(y)$ on input $(i,y)$, where $(i,t) \leftarrow
  G(n)$ and $y \ran \zo^n$, is negligible in $n$.

\end{itemize}

\noindent
On security parameter $n$ the protocol runs as follows:
\begin{enumerate}
\item

  The sender \xprot{S} has input two messages $m_0,m_1 \in \zo$.

\item

  The receiver \xprot{R} has input a choice bit $c \in \zo$.

\item

  \xprot{R} samples $(i_c,t_c) \leftarrow G(n)$ and $i_{1-c} \ran
  \zo^{\vert i_c \vert}$ and sends $(i_0,i_1)$ to \xprot{S}.

\item

  \xprot{S} samples $x_0,x_1 \ran \zo^n$ and sends $(p_{i_0}(x_0),
  H(x_0) \oplus m_0)$ and $(p_{i_1}(x_1), H(x_1) \oplus m_1)$, where
  $H$ is a (possibly randomized) hard-core bit for $p$.

\item

  \xprot{R} uses $t_c$ to compute $m_c = H(p_{i_c}^{-1}(p_{i_c}(x_c)))
  \oplus (H(x_c) \oplus m_c)$.

\end{enumerate}

It is straight-forward to prove that this protocol is computationally
secure against a static semi-honest adversary in the stand-alone
model\cite{Canetti00}: The security for the receiver is perfect, and
the receiver picks $i_{1-c}$ as to not learn $t_{1-c}$ and hence
$H(x_{1-c}) \oplus m_{1-c}$ hides $m_{1-c}$ in the sense of semantic
security.

On the other hand it is clear that the protocol is not secure against
a specious adversary: A specious adversary runs the protocol honestly,
except that it prepares $i_{1-c}$ by sampling $(i_{1-c},t_{1-c})
\leftarrow G(n)$ and then uses $t_{1-c}$ to learn $m_{1-c}$. The view
simulator $V$ adds $i_{1-c}$ to the random string $r$ such that an
execution of \xprot{R} on $r$ samples the uniformly random $i_{1-c}
\ran \zo^\ell$.

\begin{theorem}
  There exists a function $f$ and a protocol $\pi$ such that $\pi$ is
  a perfectly secure implementation of $f$ in the OT hybrid model
  against a static, unbounded semi-honest adversary, but insecure against a
  static, poly-time specious adversary.
\end{theorem}
\begin{proof}
  We look at a function $(a,b) \mapsto (x,y)$. Let $a$ be a bit, let
  $b = (b_0,b_1)$ be two bits, and let $x = b_a$ and $y =
  \epsilon$. Consider the following protocol $\pi$: it contains two
  applications of OT, where in both $\bprot$ will offer input
  $(b_0,b_1)$ and where in both $\aprot$ will input $a$. At the end
  $\aprot$ outputs $b_a$.

  It is trivial that $\pi$ is perfectly secure against a semi-honest
  adversary. It is, on the other hand, also clear that $\pi$ is not
  secure against a specious adversary, as $\ada$ can use selection bit
  $1-a$ in the second OT to learn $b_{1-a}$ and then output
  $(b_0,b_1)$. In the transcript $\alpha$ of received messages the
  view simulator $V$ simply replaces $b_{1-a}$ by $b_a$ as the message
  received from the second OT, so $\ada$ is indeed specious. It is
  also clear that no simulator for the ideal model (even if it was
  allowed active corruptions) can always output both $b_0$ and $b_1$.\qed
\end{proof}}{}

\fullv{
\section{Proof of Theorem~\ref{noswap}}\label{proofnoswap}
Suppose that there exists an $\varepsilon$-correct, $\varepsilon$-private protocol in the bare model for $\swap$ for sufficiently small $\varepsilon$; we will show that this implies that one of the two players must \emph{lose} information upon receiving a message, which is clearly impossible.

We will consider the following particular pure input state: $\ket{\varphi} := \ket{\Psi_{0,0}}^{\hreg{A}_0 \hreg{R}_{\hreg{A}}} \otimes \ket{\Psi_{0,0}}^{\hreg{B}_0 \hreg{R}_{\hreg{B}}}$, a maximally entangled state between $\hreg{A}_0 \otimes \hreg{B}_0$ and the reference system $\hreg{R}_{\hreg{A}} \otimes \hreg{R}_{\hreg{B}}$ that is broken down into two subsystems for convenience. Furthermore, we will consider the ``purified'' versions of the honest players for this protocol; in other words, we will assume that the super-operators $\aprot_1,\dots,\aprot_n$ and $\bprot_1,\dots,\bprot_n$ are in fact linear isometries and that therefore the players never discard any information unless they have to send it to the other party. The global state $\hstatei{i}{\varphi}$ after step $i$ is therefore a pure state on $\hreg{A}_i \otimes \hreg{B}_i \otimes \hreg{R}_{\hreg{A}} \otimes \hreg{R}_{\hreg{B}}$. 

After step $i$ of the protocol (i.e., after the $i$th message has been sent), Alice's state must either depend only on her own original input (if $q_i = 0$ for her simulator), or on Bob's original input (if $q_i=1$). More precisely, by the definition of privacy (Definition \ref{privdef}), we have that
\[
\Delta\left( \view{\aprot}{\varphi}{i}, \tr_{\hreg{B}_i}[\hstatei{i}{\varphi}] \right) \leq \varepsilon \enspace ,
\]
where $\view{\aprot}{\varphi}{i}$ is \aprot's simulated view after step $i$ and $\hstatei{i}{\varphi}$ is the global state in the real protocol after step $i$. Now, suppose that $q_i=0$, and let $\ket{\xi} \in \hreg{A}_i \otimes \hreg{R}_{\hreg{A}} \otimes \hreg{R}_{\hreg{B}}' \otimes \hreg{Z}$ be a purification of $\view{\aprot}{\varphi}{i}$ with $\hreg{Z}$ being the purifying system, and $\hreg{R}_{\hreg{B}}$ renamed for upcoming technical reasons. The pure state $\ket{\xi} \otimes \ket{\Psi_{0,0}}^{\hreg{R}_{\hreg{B}} \hreg{B}_0}$ has the same reduced density matrix as $\view{\aprot}{\varphi}{i}$ on $\hreg{A}_i \otimes \hreg{R}_{\hreg{A}} \otimes \hreg{R}_{\hreg{B}}$. Hence, by Uhlmann's theorem, there exists a linear isometry $V : \hreg{B}_i \rightarrow \hreg{B}_0 \otimes \hreg{Z} \otimes \hreg{R}_{\hreg{B}}'$ such that
\[
V \view{\aprot}{\varphi}{i} V^\dagger  = \ket{\xi}\bra{\xi} \otimes \ket{\Psi_{0,0}}\bra{\Psi_{0,0}}^{\hreg{B}_0\hreg{R}_{\hreg{B}}}
\]
and hence
\[
\Delta\left( V \hstatei{i}{\varphi} V^\dagger, \ket{\xi}\bra{\xi} \otimes \ket{\Psi_{0,0}}\bra{\Psi_{0,0}}^{\hreg{B}_0\hreg{R}_{\hreg{B}}} \right) \leq \sqrt{2\varepsilon}
\enspace .
\]
This means that if $q_i = 0$, then Bob is still capable of reconstructing his own input state after step $i$ by applying $V$ to his working register. Clearly, this means that $q_i' = 0$ (i.e., Bob's simulator must also not call $\swap$), and therefore, by the same argument, Alice must also be able to reconstruct her own input with an isometry $V_A : \hreg{A}_i \rightarrow \hreg{B}_0 \otimes \hreg{Z} \otimes \hreg{R}'_{\hreg{A}}$. The same argument also holds if $q_i = 1$: we then conclude that $q_i' = 1$ and that Alice and Bob must have each other's inputs; no intermediate situation is possible. We conclude that, at every step $i$ of the protocol, $q_i = q'_i$.

Now, before the protocol starts, Alice must have her input, and Bob must have his, hence, $q_0 = q_0' = 0$. At the end, the two inputs must have been swapped, which means that $q_n = q'_n = 1$; there must therefore be a step $k$ in the protocol after which the two inputs are swapped but not before, meaning that $q_k = 1$ and $q_{k-1} = 0$. But at each step, only one player receives information, which means that at this step $k$, the player who received the message must lose the ability to reconstruct his own input, which is clearly impossible. \qed}

\fullv{\section{The Rushing Lemma}\label{tensor}


Specious adversaries are guaranteed to get the correct output state
after the execution of a correct protocol. This implies that at the end of the protocol,
any extra working registers (used to implement its attack) of
any specious adversary are independent of the joint input state
of the computation. 
In other words, no extra information is available to the adversary
at the very end of the protocol. If the adversary can break the privacy
of a protocol for the two party evaluation of unitaries then it must do so 
before the last step. The adversary
must therefore rush to break privacy before the protocol ends.

\begin{lemma}[Rushing Lemma]\label{rushing}
	Let $\puq{U} = (\aprot,\bprot, n)$ be a correct protocol for the two party evaluation of $U$.  Let \ada\ be any $\varepsilon$--specious adversary in $\puq{U}$. Then, there exist an isometry $T : \tilde{\hreg{A}}_n \rightarrow \hreg{A}_n \otimes \widehat{\hreg{A}}$ and a mixed state $\tilde{\varrho} \in \dens{\widehat{\hreg{A}}}$ such that for all joint input states $\rhoin \in \dens{\hreg{A}_0\otimes\hreg{B}_0 \otimes \hreg{R}}$,
\begin{equation}\label{rusha}
	\Delta\left((T \otimes \Idt_{\hreg{B}_n \otimes \hreg{R}})\left([\ada\compose\bprot](\rhoin)\right) (V^{\dagger} \otimes \Idt_{\hreg{B}_n \otimes \hreg{R}}),\tilde{\varrho}\otimes (U \otimes \Idt_{\hreg{R}})\rhoin (U^{\dagger} \otimes \Idt_{\hreg{R}}) \right)\leq 12 \sqrt{2\varepsilon}.
\end{equation} 
The same also applies to any $\varepsilon$--specious adversary $\adb$: there exists a $T : \tilde{\hreg{B}}_n \rightarrow \hreg{B}_n \otimes \widehat{\hreg{B}}$ and a $\tilde{\varrho} \in \dens{\widehat{\hreg{B}}}$ such that
\begin{equation}\label{rushb}
	\Delta\left((T \otimes \Idt_{\hreg{A}_n \otimes \hreg{R}})\left([\aprot\compose\adb](\rhoin)\right) (V^{\dagger} \otimes \Idt_{\hreg{A}_n \otimes \hreg{R}}),\tilde{\varrho}\otimes (U \otimes \Idt_{\hreg{R}})\rhoin (U^{\dagger} \otimes \Idt_{\hreg{R}}) \right)\leq 12 \sqrt{2\varepsilon},
\end{equation} 
for every $\rhoin$.
\end{lemma}
\begin{proof}
	We shall only prove the statement for an $\varepsilon$--specious \ada; the statement for an $\varepsilon$--specious \adb\ is identical. Furthermore, by convexity, it is sufficient to prove the theorem for pure $\rhoin$.

Consider any pair of pure input states $\ket{\psi_1}$ and $\ket{\psi_2}$ in $\hreg{A}_0 \otimes \hreg{B}_0 \otimes \hreg{R}$. Now, let $\hreg{R}' := \hreg{R} \otimes \hreg{R}_2$, where $\hreg{R}_2 = \Span\{\ket{1}, \ket{2}\}$ represents a single qubit, and define the state $\ket{\psi} := \frac{1}{\sqrt{2}}(\ket{\psi_1}\ket{1} + \ket{\psi_2}\ket{2}) \in \hreg{A}_0 \otimes \hreg{B}_0 \otimes \hreg{R}'$. Note that $\trace[\hreg{R}_2]{\proj{\psi}} = \frac{1}{2} \ket{\psi_1}\bra{\psi_1} + \frac{1}{2} \ket{\psi_2}\bra{\psi_2}$. Due to the correctness of the protocol and to the speciousness of $\ada$, there exists a quantum operation $\trans{n} : \linop{\tilde{\hreg{A}_n}} \rightarrow \linop{ \hreg{A}_n}$ such that
\begin{equation*}
	\Delta\left( (\trans{n} \otimes \Idt_{\linop{\hreg{B}_n \otimes \hreg{R}'}})([\ada \compose \bprot](\ket{\psi}\bra{\psi})) , (U \otimes \Idt_{\hreg{R}'}) \ket{\psi}\bra{\psi} (U \otimes \Idt_{\hreg{R}'})^\dagger \right) \leq 2\varepsilon.
\end{equation*}
Now, consider any isometry $T : \tilde{\hreg{A}}_n \rightarrow \hreg{A}_n \otimes \widehat{\hreg{A}}$ such that $\trans{n}(\sigma) = \trace[\widehat{\hreg{A}}]{T \sigma T^\dagger}$ for every $\sigma \in \linop{\tilde{\hreg{A}}_n}$ --- in other words, any operation that implements $\trans{n}$ while keeping any information that would otherwise be destroyed in $\widehat{\hreg{A}}$. By Uhlmann's theorem, there must exist a state $\tilde{\varrho} \in \dens{\widehat{\hreg{A}}}$ such that
\begin{equation*}
	\Delta\left((T \otimes \Idt_{\hreg{B}_n \otimes \hreg{R}'})\left([\ada\compose\bprot](\ket{\psi}\bra{\psi})\right) (T^{\dagger} \otimes \Idt_{\hreg{B}_n \otimes \hreg{R}'}),\tilde{\varrho}\otimes (U \otimes \Idt_{\hreg{R}'}) \ket{\psi}\bra{\psi} (U^{\dagger} \otimes \Idt_{\hreg{R}'}) \right)\leq 2 \sqrt{2 \varepsilon}.
\end{equation*}
Now, the trace distance is monotonous under completely positive, trace non-increasing maps. In particular, we can apply the projector $P_1 = \Idt_{\linop{\hreg{A}_n \otimes \hreg{B}_n \otimes \hreg{R}}} \otimes \ket{1}\bra{1}$ to both states in the above trace distance and the inequality will still hold. In other words, we project both states onto $\ket{1}$ on $\hreg{R}_2$, thereby turning $\ket{\psi}\bra{\psi}$ into $\frac{1}{2} \ket{\psi_1}\bra{\psi_1}$. Factoring out the $\frac{1}{2}$, we get that
\begin{equation*}
	\Delta\left((T \otimes \Idt_{\hreg{B}_n \otimes \hreg{R}})\left([\ada\compose\bprot](\ket{\psi_1}\bra{\psi_1})\right) (T \otimes \Idt_{\hreg{B}_n \otimes \hreg{R}})^{\dagger},\tilde{\varrho}\otimes (U \otimes \Idt_{\hreg{R}}) \ket{\psi_1}\bra{\psi_1} (U^{\dagger} \otimes \Idt_{\hreg{R}}) \right)\leq 4 \sqrt{2 \varepsilon}.
\end{equation*} 
Likewise, projecting onto $\ket{2}$ yields
\begin{equation*}
	\Delta\left((T \otimes \Idt_{\hreg{B}_n \otimes \hreg{R}})\left([\ada\compose\bprot](\ket{\psi_2}\bra{\psi_2})\right) (T\otimes \Idt_{\hreg{B}_n \otimes \hreg{R}})^{\dagger},\tilde{\varrho}\otimes (U \otimes \Idt_{\hreg{R}}) \ket{\psi_2}\bra{\psi_2} (U^{\dagger} \otimes \Idt_{\hreg{R}}) \right)\leq 4 \sqrt{2 \varepsilon}.
\end{equation*} 
Our only problem at this point is that $\tilde{\varrho}$ in principle depends on $\ket{\psi_1}$ and $\ket{\psi_2}$. However, repeating the above argument with $\ket{\psi_1}$ and $\ket{\psi_3}$ for any $\ket{\psi_3}$ will yield a $\tilde{\varrho}'$ with 
\begin{equation*}
	\Delta\left((T \otimes \Idt_{\hreg{B}_n \otimes \hreg{R}})\left([\ada\compose\bprot](\ket{\psi_1}\bra{\psi_1})\right) (T^ \otimes \Idt_{\hreg{B}_n \otimes \hreg{R}})^\dagger,\tilde{\varrho}'\otimes (U \otimes \Idt_{\hreg{R}}) \ket{\psi_1}\bra{\psi_1} (U^{\dagger} \otimes \Idt_{\hreg{R}}) \right)\leq 4 \sqrt{2 \varepsilon}
\end{equation*} 
and hence, by the triangle inequality, $\Delta(\tilde{\varrho}, \tilde{\varrho}') \leq 8 \sqrt{2\varepsilon}$. Therefore, for any state $\ket{\varphi} \in \hreg{A}_0 \otimes \hreg{B}_0 \otimes \hreg{R}$, there exists a state $\tilde{\rho} \in \widehat{\hreg{A}}$ with $\Delta(\tilde{\rho}, \tilde{\varrho}) \leq 8 \sqrt{2\varepsilon}$ such that
\begin{equation*}
	\Delta\left((T \otimes \Idt_{\hreg{B}_n \otimes \hreg{R}})\left([\ada\compose\bprot](\ket{\varphi}\bra{\varphi})\right) (T \otimes \Idt_{\hreg{B}_n \otimes \hreg{R}})^{\dagger},\tilde{\rho}\otimes (U \otimes \Idt_{\hreg{R}}) \ket{\varphi}\bra{\varphi} (U^{\dagger} \otimes \Idt_{\hreg{R}}) \right)\leq 4 \sqrt{2 \varepsilon}.
\end{equation*} 
The lemma then follows by the triangle inequality:
\begin{multline*}
	\Delta\left((T \otimes \Idt_{\hreg{B}_n \otimes \hreg{R}})\left([\ada\compose\bprot](\ket{\varphi}\bra{\varphi})\right) (T \otimes \Idt_{\hreg{B}_n \otimes \hreg{R}})^{\dagger},\tilde{\varrho}\otimes (U \otimes \Idt_{\hreg{R}}) \ket{\varphi}\bra{\varphi} (U^{\dagger} \otimes \Idt_{\hreg{R}}) \right)\\
	\begin{split}
		&\leq \Delta\left((T \otimes \Idt_{\hreg{B}_n \otimes \hreg{R}})\left([\ada\compose\bprot](\ket{\varphi}\bra{\varphi})\right) (T \otimes \Idt_{\hreg{B}_n \otimes \hreg{R}})^{\dagger},\tilde{\rho}\otimes (U \otimes \Idt_{\hreg{R}}) \ket{\varphi}\bra{\varphi} (U^{\dagger} \otimes \Idt_{\hreg{R}}) \right) + \Delta(\tilde{\rho}, \tilde{\varrho})\\
		&\leq 4 \sqrt{2\varepsilon} + 8 \sqrt{2\varepsilon} = 12 \sqrt{2\varepsilon}
	\end{split}
\end{multline*}
\qed
\end{proof}}{}


%

\fullv{\section{Proof of Privacy}\label{privacyproof}

\newcommand{\sbc}{\ensuremath{S_{\mbox{\sc bc}}}}
\newcommand{\xbc}{\ensuremath{X_{\mbox{\sc bc}}}}
\newcommand{\badphiin}{\ensuremath{\ket{\varphi^*_{\tt in}}}}

In the following we prove the privacy 
 $\pup{U}^{\qop{O}}=(\aprots,\bprots,n_U+1)$
against 
specious quantum adversaries and that 
for any unitary $U\in \unit{\hreg{A}_0\otimes\hreg{B}_0}$ 
represented by a quantum
circuit $C_U$ with $u$ gates in \ug.
We provide families of simulators $\simul{\ada}$ and $\simul{\adb}$
for any specious quantum adversaries $\ada$ and $\adb$ respectively.
Since the protocol has $n_U$ oracle calls, it is sufficient to provide
simulators for each of these $n_U$ steps since the final 
quantum operations (i.e., \alice{n_U+1}\ and \bob{n_U+1})
are local.
No simulator for a round occurring before 
the start of the {\em key-releasing} phase
needs to call the ideal functionality for $U$. 
The output of these simulators will be shown identical
to the adversary's view (i.e., the simulation is perfect)
even if the adversary is arbitrarily malicious. 
Only the last
simulator of each family  
needs to call to the ideal functionality for $U$. 
The last simulation produces a state that is essentially
$\sqrt{\varepsilon}$--close
to adversary's view provided it is $\varepsilon$--specious.


First, we show privacy of the {\em evaluation phase}
before addressing privacy of the {\em key-releasing phase}.
Privacy of the entire protocol will then follow.







\subsection{Privacy of the Evaluation Phase}\label{sect:cliffordprivate}

We start by showing privacy of protocol $\pup{U}^{\qop{O}}=(\aprots,\bprots,n_U+1)$
 at all steps $1\leq i \leq n_U-1$ occurring during 
the {\em evaluation  phase}
of quantum  
circuit $C_U$ implementing $U$ with $u$ gates in \ug. 
The last step of the evaluation phase is $n_U-1$ since only
one oracle call is left to complete the execution.
This phase is the easy part of the simulation since 
all transmissions  are
independent of the joint input state $\rhoin \in \dens{\hreg{A}_0\otimes\hreg{B}_0\otimes \hreg{R}}$.
The theorem below provides a perfect simulation of  any adversary's view generated
during the evaluation of any gate in $C_U$. No call to the ideal functionality for $U$
is required.

\begin{theorem}[Privacy of the Evaluation]\label{induct} 
$\pup{U}^{\qop{O}}=(\aprots,\bprots,n_U+1)$
admits simulators $\simulator{\ada}$
and $\simulator{\adb}$ 
that do not
call the ideal functionality for $U\in \unit{\hreg{A}_0 \otimes \hreg{B}_0}$ such that
for any joint input state $\rhoin \in \dens{\hreg{A}_0 \otimes\hreg{B}_0\otimes \hreg{R}}$, 
every $1\leq i \leq n_U-1$:
\begin{equation}\label{ss1}
\Delta\left(\view{\ada}{\rhoin}{i},\trace[\hreg{B}_i]{\astatei{i}{\ada}{\rhoin}} \right) = 0 
\mbox{ and }
\Delta\left(\view{\adb}{\rhoin}{i},\trace[\hreg{A}_{i}]{\astatei{i}{\adb}{\rhoin}} \right) = 0,
\end{equation}
This  holds against any adversaries \ada\ and \adb, not necessarily specious.
\end{theorem}
\begin{proof}[Sketch]
Without loss of generality, we prove privacy
only against adversary $\ada$. The protocol
being symmetric,
privacy against $\adb$ follows.
The proof proceeds by induction on the current 
gate $G_j$ in the circuit $C_U:= G_{u} G_{u-1}  \ldots G_1$
evaluated in $\pup{U}^{\qop{O}}$.
We provide simulators $\simulg{j}$ producing
\ada's view after the evaluation of $G_j$. 
During the execution of $G_j$, \ada\ may receive
at most one message from \bprot\  and may call
the ideal \nl\ at most once (when $G_j=\rgate$).
It means that during the evaluation of $G_j$, 
no, one, or two simulations will be  needed since it consists
in no, one or two oracle calls out of which at most one
is non-trivial.
Let $s[j]\in \{0,1,2\}$ for $1\leq j\leq u$  be the number of steps
to be simulated 
during the evaluation of $G_j$. Let $i[0]:=0$ and
$i[j] = s[j]+i[j-1]$ for $1\leq j \leq u$ be all steps to simulate for the evaluation
of $G_j G_{j-1}\cdots G_{1}$.  In order to fulfill privacy as defined in Definition~\ref{privacy},
each simulator
\simulg{j}\  must be converted  into 
$\simul{i[j-1]+1}\in \simulator{\ada}$ if $i[j]=i[j-1]+1$ (i.e., $G_j$ requires
only one step to be simulated and this step is a message from \bprots) 
and into $\{\simul{i[j-1]+1}, \simul{i[j-1]+2}\}\subseteq \simulator{\ada}$
if $i[j]=i[j-1]+2$ (i.e., $G_j=\rgate$ and therefore requires two simulation steps: one message
from \bprots\ and one call to \nl). This conversion is performed the following way.
We let $\simul{i[j-1]+1}$
run
\simulg{j}\  until the $i[j-1]+1$--th step is reached. This step
is necessarily a message transmitted from \bprots\ to
\ada. 
If $i[j]=i[j-1]+2$ then 
$\simul{i[j]} := \simulg{j}$  which corresponds to the simulation
of \ada's  view  after the call to \nl. We now provide
$\simulg{j}$ for each gate $G_j$ in $C_U$. Notice
that we do not explicitly simulates a communication from
\ada\ to \bprots\ since simulating this step can be performed
from the simulation of the previous step together with \ada\
quantum operations at the current step.


%


$\simulg{1}$ works as follows. It  runs \ada\ on her part of the joint input state 
$\rhoin\in \dens{\hreg{A}_0\otimes \hreg{B}_0\otimes \hreg{R}}$ until the first message from the other
party is expected. If  gate $G_1$ does not involve any transmission from \bprots\
then the simulation of gate $G_1$ is over (i.e., $G_1$ is in $\{X,Y,Z, \hgate,\pgate\}$
or a \cnot\ applied on local wires). Otherwise,
it prepares the first message sent from \bprots. Of course, this message
depends on $G_1$.  We have the following three cases to address:
\begin{description}
\item[\cnot-gate:] \ada\ holds the {\em target wire} while \bprots\ 
holds the {\em control wire}. This case is the only one where \ada\ receives
something from \bprots\ during the computation of a \cnot-gate. $\simulg{1}$ then works
the obvious way in order to generate the first transmission from \bprots\ to \ada:
\begin{itemize}
\item $\simulg{1}$ prepares
 $\ket{\xi} = (\Idt_2 \tensor \cnot \tensor \Idt_2) \ket{\Psi_{0,0}}^{\hreg{W}}
\tensor \ket{\Psi_{0,0}}^{\hreg{A}^{\qop{O}}_1}$ where \hreg{W}\ is a new working register for the simulator.
\simulg{1} then sends register $\hreg{A}^{\qop{O}}_1$  to \ada. This simulates
\bprots's transmission to \ada.
\item The transmission prepared by $\simulg{1}$ 
      is in the same state as when \ada\ interacts with \bprots\
      upon any input state \rhoin.
       It follows that 
      the output of  $\simulg{1}$ satisfies:
  \[  \Delta\left(\view{\ada}{\rhoin}{1},\trace[\hreg{B}_1]{\astatei{1}{\ada}{\rhoin}} \right) = 0 ,
  \]  
      for all input states $\rhoin \in \dens{\hreg{A}_0\otimes\hreg{B}_0\otimes \hreg{R}}$.
\end{itemize}
\item[\rgate-gate:] \ada\ holds the register upon which
       the gate is executed. In this case, $\simulg{1}$ provides \ada\ with
        \bprots's  as follows:
\begin{itemize}
\item $\simulg{1}$ prepares and sends $\Idt_2\in \dens{\hreg{A}^{\qop{O}}_1}$ to \ada. 
%
\item \simulg{1}\ then call the ideal \nl\ with a random input bit. Notice that \ada\ cannot
distinguish this behavior from  \bprots's since an AND-box is non-signaling and can therefore
not be used by one party to extract any information about the other party's input state (i.e., the output
of one party can be generated before the input of the other party has been provided to the box).
\item As in the case where \ada\ interacts with \bprots, the first
       message received  from $\simulg{1}$ is in state $\Idt_2$ and \ada's output 
       of \nl\ is independent of \bprot's view. It follows
       that,
  \[  \Delta\left(\view{\ada}{\rhoin}{1},\trace[\hreg{B}_1]{\astatei{1}{\ada}{\rhoin}} \right) = 0,
  \]  
       for all input states $\rhoin\in\dens{\hreg{A}_0\otimes\hreg{B}_0\otimes\hreg{R}}$. 
\end{itemize}
\item[\rgate-gate:] \bprots\  holds the register upon which
       the gate is executed.  $\simulg{1}$ provides \ada\ with \bprots's first transmission
       the same way than in the previous case:
 \begin{itemize}
       \item \simulg{1}\ prepares and  sends 
             $\Idt_2\in \dens{\hreg{A}^{\qop{O}}_1}$ to  \ada.  This simulates \bprots\ transmission
             to \ada.
       \item \simulg{1}\ provides the \nl\ with a fresh random bit as for in the previous case.
       \item As in the case where \ada\ interacts with \bprots, the first
            message received  from $\simulg{1}$ is in state $\Idt_2$ and \ada's output 
             of \nl\ is independent of \bprots's view. It follows
             that,
             \[  \Delta\left(\view{\ada}{\rhoin}{1},\trace[\hreg{B}_1]{\astatei{1}{\ada}{\rhoin}} \right) = 0,
               \]  
               for all input states $\rhoin\in\dens{\hreg{A}_0\otimes\hreg{B}_0\otimes\hreg{R}}$.      
  \end{itemize}  
\end{description} 
 Since the three cases above exhaust all possibilities for a transmission
 from \bprot\ to \ada, $\simulg{1}$ satisfies
(\ref{ss1}).

Now, suppose by the induction hypothesis that $\simulg{j-1}$ simulates perfectly up to and including the
$j-1$-th step of the adversary \ada\  for $2\leq j\leq n_U-1$. 
We now show how to construct $\simulg{j}$ simulating  perfectly
up to an including gate $G_j$.
We construct $\simulg{j}$ the obvious
way.  $\simulg{j}$ runs
$\simulg{j-1}$ and then simulates  \ada\ until \bprots's next transmission occurring during
the evaluation of $G_j$. If no such message occurs during the evaluation of $G_j$
then $\simulg{j}$ is done. Otherwise, 
the same  three cases described above have to be considered.
$\simulg{j}$ provides \ada\ with  \bprots's  transmission exactly
the same way than for $\simulg{1}$.  The result follows easily.
%
\qed
\end{proof}



\subsection{Privacy of the Key-Releasing Phase}
In order to conclude the privacyof $\pup{U}^{\qop{O}}$, families $\simul{\ada}$
and $\simul{\adb}$  
need one more simulator each: $\simul{n_U}\in  \simulator{\ada}$
and $\simul{n_U}' \in  \simulator{\adb}$  corresponding
to the simulation of the key-releasing phase. 
This time, these simulators 
need to query the ideal functionality for $U$ and also
need the adversary to be specious.  
We show that privacy of the key-releasing phase 
follows from the ``Rushing Lemma'' (Lemma~\ref{rushing}).
The lemma tells us that as soon as the output is available 
to the honest player, it is too late for specious adversaries to break
privacy. This is the role of the ideal \swap\ to make sure that 
before the adversary  gets the output of the computation, the information
needed by the honest player to recover its own output has been given away
by the adversary. 

It should be mentioned that we're not explicitly simulating the 
final state of the adversary since simulating the
\swap\ allows also to get \ada's final state  
by simply adding \ada's last quantum operation to the simulated
view. We therefore set step $n_U$ in $\pup{U}^{\qop{O}}$
to be the step reached after the call to \swap. This abuses the notation
a  bit since after \swap, \ada\ and \bprots\ must each apply 
a final quantum operation with no more oracle call. We'll denote
by $\ada_{n_U+1}$ and $\bprots_{n_U+1}$ these last operations
allowing to reconstruct the output of the computation (no comunication).

\begin{lemma}\label{release}
For any $\varepsilon$-specious quantum adversaries \ada\
and \adb\ 
against $\pup{U}^{\qop{O}}=(\aprots,\bprots,n_U+1)$, there 
exist  simulators $\simul{n_U}\in \simulator{\ada}$ and
$\simul{n_U}'\in \simulator{\adb}$ such that for all 
$\rhoin \in \dens{\hreg{A}_0\otimes \hreg{B}_0\otimes\hreg{R}}$,
\begin{equation}\label{ss2}
\begin{split}
\Delta\left(\view{\ada}{\rhoin}{n_U},\trace[\hreg{B}_{n_U}]{\astatei{n_U}{\ada}{\rhoin}} \right) \leq  
24\sqrt{2\varepsilon} 
\mbox{\,\,\, and \hspace{2in}} & \\
\Delta\left(\view{\adb}{\rhoin}{n_U},\trace[\hreg{A}_{n_U}]{\astatei{n_U}{\adb}{\rhoin}} \right) \leq
24\sqrt{2\varepsilon}\, .
\end{split}
\end{equation}
Simulators $\simul{n_U}$ and $\simul{n_U}'$ call the ideal functionality for $U$
and imply the simulations of step $n_U+1$ as well.
\end{lemma}
\begin{proof}[sketch]
Once again, we only prove privacy against adversary \ada. 
The privacy against \adb\ follows directly since the key-releasing phase
is completely symmetric.  The idea behind the proof is to run 
\ada\ and \bprots\ upon a dummy joint input state until the end
of the protocol. Since the adversary is specious, it can re-produce
the honest state at the end. The Rushing Lemma tells us that
at this point, the output of the computation is essentially in tensor
product with all the other registers. Moreover, the state of all other
registers is independent of the input state upon which the protocol
is executed. The {\em dummy output} can then be replaced
by the output of the ideal functionality for $U$ before \ada\
goes back to the stage reached just after \swap. 


More formally,
we define a simulator $\simul{n_U}\in \simulator{\ada}$ producing 
\ada's view just after the call to \swap. 
Let $\ada_{\swap}\in \lin{\hreg{A}_0}{\tilde{\hreg{A}}_{n_U}}$ and 
$\bprots_{\swap}\in \lin{\hreg{B}_0}{\tilde{\hreg{B}}_{n_U}}$ be the 
quantum operations run by \ada\ and \bprots\  respectively until \swap\ is executed.
Notice that at this point, \ada's and \bprots's registers
do not  have any further oracle registers  
since no more communication or oracle call will take place.
Let $\tilde{A}_{n_U}\in\lin{\tilde{\hreg{A}}_{n_U}}{\tilde{\hreg{A}}_{n_U+1}\otimes \hreg{Z}}$ be the isometry 
implementing \ada's  last quantum operation  taking
place after the call to \swap\ (and producing the outcome) and
let   ${B}_{n_U}\in\lin{{\hreg{B}}_{n_U}}{{\hreg{B}}_{n_U+1}\otimes \hreg{W}}$ be the isometry 
implementing \bprots's  last quantum operation.  Finally,
let $T\in \lin{\tilde{\hreg{A}}_{n_U+1}}{\hreg{A}_{n_U+1}\otimes\hat{\hreg{A}}}$ be the isometry
implementing $\trans{n_U+1}$ as defined in Lemma~\ref{rushing}
(i.e., the transcript produced at the very end of the protocol). 
As usual , let $\rhoin \in \dens{\hreg{A}_0\otimes\hreg{B}_0\otimes\hreg{R}}$ 
be the joint input state of $\pup{U}^{\qop{O}}$. 
The simulator $\simul{n_U}$  performs the following operations:

\begin{enumerate}
\item $\simul{n_U}$ generates 
the quantum state $\sigma(\phi^*)=[\ada_{\swap}\compose\bprots_{\swap}](\proj{\phi^*})
\in \dens{\tilde{\hreg{A}}_{n_U}\otimes \hreg{B}_{n_U}}$
 implementing 
\ada\ interacting with \bprots\ until \swap\  is applied. The execution
is performed  upon a predetermined (dummy) 
arbitrary input state
$\ket{\phi^*} \in \hreg{A}_0\otimes\hreg{B}_0$. \label{s1}

\item $\simul{n_U}$ sets $\sigma'(\phi^*) = (T \tilde{A}_{n_U}\otimes {B}_{n_U})\cdot
\sigma(\phi^*) \in \dens{\hreg{A}_{n_U+1}\otimes \hreg{B}_{n_U+1}
\otimes \hreg{Z}\otimes\hat{\hreg{A}}\otimes\hreg{W}}$.\label{s2}

\item \simul{n_U}\ replaces register $\hreg{A}_{n_U+1}\approx \hreg{A}_0$ 
by \aprots's output
of the ideal functionality for $U$ evaluated upon \rhoin.
That is, \simul{n_U}\ generates the state
$\sigma'(\rhoin)= (U\otimes \Idt_{\hreg{R}}) \rhoin (U\otimes\Idt_{\hreg{R}})^{\dagger} 
\otimes \trace[\hreg{A}_{n_U+1}\hreg{B}_{n_U+1}]{\sigma'(\phi^*)}
\in \dens{\hreg{A}_{n_U+1}\otimes \hreg{B}_{n_U+1}\otimes\hreg{R}
\otimes \hreg{Z}\otimes\hat{\hreg{A}}\otimes\hreg{W}}$.
\label{s3}

\item \simul{n_U}\ finally sets $\view{\ada}{\rhoin}{n_U}=\trace[\hreg{B}_{n_U}\hreg{W}]
{(T \tilde{A}_{n_U}\otimes \Idt_{\hreg{B}_{n_U+1}\hreg{R}})^{\dagger}\cdot \sigma'(\rhoin)}
\in \dens{\tilde{\hreg{A}}_{n_U}\otimes\hreg{R}}$.
 \label{s4}
\end{enumerate}
Notice that execution of  the ideal \swap\ ensures that 
the keys swapped are independent of each other
and of the joint input state $\rhoin$. This is because
for any input state, all these keys are uniformly distributed
bits if they are outcomes of Bell measurements and otherwise
are set to $0$.  By the Rushing Lemma~\ref{rushing} and the fact
that \ada\ is $\varepsilon$--specious,
we have:
\[ 
\begin{split}
 \Delta\left( \trace[\hreg{Z}\hat{\hreg{A}}\hreg{W}]{\sigma'(\phi^*)}, \tilde{\varrho}\otimes U\proj{\phi^*}U^\dagger\right) \leq 
12\sqrt{2\varepsilon} \mbox{ and \hspace{2in}}  &    \\ 
 \Delta\left( (\trans{n_U+1}\otimes\Idt_{\linop{\hreg{B}_{n_U+1}}}) \left([\ada\compose \bprots](\rhoin)\right), \tilde{\varrho}\otimes U\rhoin U^\dagger\right)  \leq 
12\sqrt{2\varepsilon}.
\end{split}
\]
It follows using the triangle inequality  that,
\begin{equation}\label{wow}
 \Delta\left( (\trans{n_U+1}\otimes\Idt_{\linop{\hreg{B}_{n_U+1}}}) \left([\ada\compose \bprots](\rhoin)\right),   
\trace[\hreg{Z}\hat{\hreg{A}}\hreg{W}]{\sigma'(\rhoin)}
\right) \leq 24 \sqrt{2\varepsilon}.
\end{equation}
Using the fact that isometries cannot increase the trace-norm distance and
that $(T\tilde{A}_{n_U})^{\dagger}$ allows \ada\  to go back from the end of the protocol
to the step reached after \swap, 
we get from (\ref{wow}) that:
\begin{align*}
\Delta\left(\view{\ada}{\rhoin}{n_U},\trace[\hreg{B}_{n_U}]{\astatei{n_U}{\ada}{\rhoin}} \right)
&= \Delta\left(
(\trans{n_U+1}\otimes\Idt_{\linop{\hreg{B}_{n_U+1}}}) \left([\ada\compose \bprots](\rhoin)\right),   
\trace[\hreg{Z}\hat{\hreg{A}}\hreg{W}]{\sigma'(\rhoin)}
\right) \\
&\leq 24\sqrt{2\varepsilon}.
\end{align*}
The statement follows.
\qed
\end{proof}

Theorem~\ref{induct} and Lemma~\ref{release} imply the
privacy of $\pup{U}^{\qop{O}}$ against specious
adversaries and that for any $U\in \unit{\hreg{A}_0\otimes\hreg{B}_0}$ as
stated in our main Theorem~\ref{final}.
When $U$ is in the Clifford group, one call to 
an ideal \swap\ 
is sufficient to ensure privacy. 
Unitaries in the
Clifford group are, to some extent, the {\em easy} ones
since although an ideal functionality is required for privacy,
that functionality is unitary and belongs to the Clifford group
rather than a classical cryptographic primitive. 
If $U$ is not in the Clifford
group however,  one additional  call to a classical \nl\ is required for each
\rgate-gate. This is reminiscent to classical circuits with {\sc and}
gates where oblivious transfer is required to be able to evaluate them
privately. In order to implement a classical \nl, commitments and 
quantum communication are sufficient and necessary \cite{DFLSS09,Lo97}.
}{}
\end{appendix}




\end{document}